\documentclass[a4paper,11pt]{article}
\usepackage{amsmath}
\usepackage{amssymb}
\usepackage[dvips]{graphicx}
\usepackage{graphicx}
\usepackage{fullpage}
\usepackage{appendix}
\usepackage{color}  
\usepackage{times}

\newcommand{\remove}[1]{}

\newtheorem{assumption}{Assumption}

\newtheorem{claim}{Claim}

\newtheorem{proposition}{Proposition}

\newtheorem{corollary}{Corollary}

\newtheorem{definition}{Definition}

\newtheorem{lemma}{Lemma}

\newtheorem{problem}{Problem}

\newtheorem{theorem}{Theorem}

\newenvironment{proof}{\noindent{\bf Proof\@:}}{\hfill $\Diamond$\\}

\newenvironment{theoremproof}[1]{\noindent{\bf Proof of Theorem #1\@:}}{\hfill $\Diamond$\\}

\newenvironment{claimproof}[1]{\noindent{\bf Proof of Claim #1\@:}}{\hfill $\Diamond$\\}
\newenvironment{propositionproof}[1]{\noindent{\bf Proof of Proposition #1\@:}}{\hfill $\Diamond$\\}

\newcommand{\myshow}[1]{}

\newcommand\inappendix[1]{#1}
\newcommand{\randbits}[1]{}

\date{\today}

\title{\bf A simple algorithm for random colouring $G(n, d/n)$
using $(2+\epsilon)d$ colours. \thanks{Supported by EPSRC grant EP/G039070/2 and DIMAP.}}

\author{
Charilaos Efthymiou\\
University of Warwick, Mathematics and Computer Science, Coventry CV4 7AL, UK\\
\texttt{c.efthymiou@warwick.ac.uk}
}

\begin{document}
\maketitle
\thispagestyle{empty}

\begin{abstract}

\noindent
Approximate random $k$-colouring of a graph $G=(V,E)$ is a very
well studied problem in computer science and statistical physics.
It amounts to constructing a $k$-colouring of $G$ which is distributed 
close to {\em Gibbs distribution}, i.e. the uniform distribution over 
all the $k$-colourings of $G$. Here, we deal with the problem when the
underlying graph is an instance of Erd\H{o}s-R\'enyi random graph
$G(n,p)$, where $p=d/n$ and $d$ is fixed.

We propose a novel efficient algorithm for approximate random $k$-colouring 
with the following properties: given an instance of $G(n,d/n)$ and for
any  $k\geq (2+\epsilon)d$, it returns a $k$-colouring distributed within 
total variation distance $n^{-\Omega(1)}$ from the Gibbs distribution,
with probability $1-n^{-\Omega(1)}$.

What we propose is neither a MCMC algorithm nor some algorithm inspired
by the message passing heuristics that were introduced by statistical
physicists. Our algorithm is of {\em combinatorial} nature. It is based on 
a rather simple recursion which reduces the random $k$-colouring of 
$G(n,d/n)$ to random $k$-colouring simpler subgraphs first.

The  lower bound on the number of colours for our algorithm to 
run in
polynomial time is {\em dra\-ma\-ti\-cal\-ly} smaller than the 
corresponding bounds we have for any previous algorithm.


\end{abstract}

\newpage
\setcounter{page}{1}

\section{Introduction}\label{sec:intro}	
Approximate random $k$-colouring of a graph $G=(V,E)$ is a very
well studied problem in computer science and statistical physics.
It amounts to constructing a $k$-colouring of $G$ which is distributed
close to {\em Gibbs distribution}, i.e. the uniform distribution
over all the $k$-colourings of $G$. Here, we deal with the specific
algorithmic problem when the underlying graph is an instance of
Erd\H{o}s-R\'enyi random graph $G(n,p)$, where $p=d/n$ and $d$ is
fixed. 

The most powerful and most popular algorithms for this kind of
problems are based on the Markov Chain Monte Carlo (MCMC) method.
There the main technical challenge is to establish that the
underlying Markov chain mixes in polynomial time (see \cite{jerrum}).
The work in \cite{mossel-colouring-gnp} (which improved \cite{old-GnpSampling})
shows that the well known Markov chain {\em Glauber dynamics} for
$k$-colourings has polynomial time mixing for typical instances of
$G(n,d/n)$ as long as the number of colours $k$ is larger than
$k(d)$, a number which depends only on the expected degree
of $G(n,d/n)$, $d$.

Notably, both \cite{mossel-colouring-gnp,old-GnpSampling}
overcame the ``maximum degree obstacle'' from which most techniques for 
analysing the mixing time of the Glauber dynamics suffer, i.e. 
they are stated in terms of the maximum  degree of the graph.
This makes them insufficient for $G(n,d/n)$ where the maximum 
degree grows as $\Theta((\log n)/(\log \log n))$ with probability $1-o(1)$.

Recently physicists proposed some heuristics for computing deterministically
marginals of Gibbs distribution. These heuristics are based on the
message passing algorithm Belief Propagation (see \cite{BeliefPropRigorous})
and {\em ideally} they could be used for random colouring $G(n,d/n)$. 
There, the main challenge is to show that the computation of
the marginals is accurate. In turn this amounts to establishing 
certain {\em spatial} correlation decay conditions for the Gibbs
distribution. 
We should remark that the heuristics proposed by statistical
physicists, namely {\em Belief Propagation guided decimation}
and {\em Survey Propagation guided decimation}, were put forward
on the basis of very insightful but highly non-rigorous statistical
me\-cha\-nics considerations (see \cite{SP-heuristic}). 
The work in \cite{TCS-sampling} presents an efficient algorithm which 
is a variation of Belief propagation and returns an approximate random 
$k$-colouring of a typical $G(n,d/n)$ as long as $k>k'(d)$, where $k'(d)$
is a  number which depends only the expected degree (i.e. $k'(d)=d^{14}$).

In this work we propose a novel algorithm for approximate 
random $k$-colouring $G(n,d/n)$ which  not only overcomes
the ``maximum  degree obstacle'' but somehow {\em optimizes}
the dependence of the minimum number of colours from the
expected degree $d$. The lower
bound on the number of colours for our algorithm to run in polynomial
time is {\em dra\-ma\-ti\-cal\-ly} smaller than the corresponding
bounds we have for any previous algorithm on the problem.
The algorithm does not fall into any of the previous two categories,
i.e., it is neither MCMC nor based on Gibbs marginals computation.

We are based on the following humble observation: Let
$\{v,u\}$ be an edge of $G_{n,d/n}$. A random colouring of 
$G_{n,d/n}$ can be seen as a  random colouring of $G_{n,d/n}
\backslash\{v,u\}$ with the additional pro\-per\-ty that $v$ and $u$
are assigned different colours.
Assume that we have a polynomial time algorithm, call it {\tt STEP},
such that given any graph $G$ and two non-adjacent
vertices $v$, $u$ it transforms a random colouring of $G$ to a
random  colouring which has the {\em extra} property that $v$ and
$u$ take different colours.
In that case, the initial problem can be reduced to taking a 
random $k$-colouring of $G_{n,d/n}\backslash\{v,u\}$ and 
then use {\tt STEP}. The reasonable question, then,  would be
how can someone colour $G_{n,d/n}\backslash\{v,u\}$ randomly.
We can  set up a recursion by applying the previous reduction 
for $G_{n,d/n}\backslash \{v,u\}$ and so on.
Note that as the recursion proceeds the structure of the graph 
that is considered gets simpler and simpler. This is due to the
edge deletions. Clearly, after a certain number of recursive calls
the graph becomes so simple that it can be $k$-coloured randomly
in polynomial time by some known algorithm.

A great deal of this work illustrates the implementation of 
{\tt STEP} in the special case where the input graph $G$
is a typical instance of $G_{n,d/n}$ or any of its subgraphs that
are considered in the recursion.
{\tt STEP}  will be an approximation algorithm, i.e. given a random
colouring of $G$ in the input the distribution of the output will be 
an approximation of the desired one. Consequently, at the end we get
an approximate random $k-$colouring of $G(n,d/n)$.

We use total variation distance as a measure of distance between
distributions.

\begin{definition}
For the distributions $\nu_{a}, \nu_{b}$ on $[k]^V$, let 
$|| \nu_{a}-\nu_{b}||$ denote their {\em total variation distance},
i.e.
\begin{displaymath}
|| \nu_{a}-\nu_{b}||=\max_{\Omega' \subseteq [k]^V} | \nu_{a}(\Omega')-\nu_{b}(\Omega')|.
\end{displaymath}
For $\Lambda \subseteq V$ let $||\nu_{a}-\nu_{b}||_{\Lambda}$
denote the total variation distance between the projections of 
$\nu_a$ and $\nu_{b}$ on $[k]^{\Lambda}$.
\end{definition}

\noindent
{\tt STEP} will have the following {\em general property}:
Consider in the input a random $k$-colouring of some graph $G$
and $v$, $u$, two non-adjacent vertices of $G$.
The accuracy of the outcome depends on certain {\em spatial mixing} 
properties of the Gibbs distribution of the colourings of $G$.
In particular, for a random $k$-colouring of $G$  it suffices 
that there is a sufficiently large $b>0$ such that 
\begin{equation}\label{eq:SpatialRequirement}
\left|Pr[\textrm{$u$ is coloured $c$}| \textrm{$v$ is coloured $q$}]-
\frac{1}{k}\right| 
\leq \exp\left( -b\cdot dist(v,u)\right) \qquad \forall c,q\in [k].
\end{equation}
Moreover, assuming that (\ref{eq:SpatialRequirement}) holds, then
the distribution of the output of {\tt STEP} is within total variation
distance from the ideal distribution a quantity which is proportional to
the r.h.s. of (\ref{eq:SpatialRequirement}).
Consequently, when we consider the previous recursive random colouring
algorithm (that uses {\tt STEP}), we note that it is desirable 
to delete edges that belong to long cycles in each recursive
call.

We show that for a typical $G(n,d/n)$ and for $k\geq(2+\epsilon)d$, 
where $\epsilon>0$ is fixed, we get a relation as in  (\ref{eq:SpatialRequirement})
for the random $k$-colourings of any graph in the recurrence. 
Moreover, if we are careful enough on how do we
delete the edges in the recurrence, the outcome of the random 
colouring algorithm is very close to Gibbs distribution. In particular, 
we show the following theorem.

\begin{theorem}\label{thrm:GnpAccuracy}
Let $\mu$ be the uniform distribution over the $k$-colourings of
$G_{n,d/n}$ and let $\mu'$ be the distribution of the colouring
that is returned by our random colouring algorithm. Taking $k\geq
(2+\epsilon)d$, for fixed $\epsilon>0$,  then with probability at
least $1-n^{-\frac{\epsilon}{90\log d}}$ it holds that 
\begin{displaymath}
||\mu-\mu' || = O\left(n^{-\frac{\epsilon}{90\log d}}\right),
\end{displaymath}
for any fixed $d>d_0(\epsilon)$.
\end{theorem}
%
Additionally, we provide guarantees on the time com\-ple\-xi\-ty of the 
algorithm.

\begin{theorem}\label{thrm:GnpTimeCmplxt}
With probability at least $1-n^{-2/3}$, it holds that the time complexity
of the random colouring algorithm is $O(n^{2})$.
\end{theorem}
\inappendix{
Detailed proofs of Theorem \ref{thrm:GnpAccuracy} and Theorem 
\ref{thrm:GnpTimeCmplxt} appear in the appendix, Section \ref{sec:AppGnp}.
}
\vspace{-.4cm}

\remove{
To get an idea of the achievements in the previous works consider the
following. It is well known that the maximum degree of $G(n,d/n)$ is
much larger than the average degree $d$, i.e.  with probability $1-o(1)$,
the maximum degree is $\Theta\left(\frac{\log n}{\log\log n} \right)$
(see \cite{janson}).
The main difficulty in dealing with high degree vertices in the analysis
is that in many colourings, these vertices have few colour choices, i.e., 
almost  all of the colours might appear in their neighbourhood. Thus, the 
colour choice of the neighbour of a high degree vertex $v$ can have a large
influence on the colour choice of $v$.

}

\randbits{
It is straightforward to show that there cannot be any algorithm 
that  $k$-colours randomly any graph on $n$ vertices with less 
than $\Theta(n)$ random bits, when $k\geq 2$. 
It follows easily from the analysis of the algorithm that 
the order of magnitude of the number of random bits we require
 matches the linear lower bound.

\begin{corollary}\label{cor:RandomBits}
With probability at least $1-n^{-2/3}$ the number of random bits
that the random colouring algorithm requires is $\Theta(n)$.
\end{corollary}
}

\paragraph{Notation} We denote with small letters of the greek
alphabet  the colourings of a graph $G$, e.g. $\sigma, \eta, \tau$,
while we use capital letters for the random variables which take
values over the colourings e.g. $X,Y, Z$. 
We denote with $\sigma_v$ the colour assignment of the vertex $v$
under the colouring $\sigma$. Similarly, the random variable $X(v)$
is is equal to the colour assignment that $X$ specifies for the
vertex $v$. Finally, for an integer $k>0$ let $[k]=\{1,\ldots,k\}$.

\section{Basic Description}\label{sec:basics}
In this section we provide a more detailed description of our
approximate colouring algorithm. 
We assume that the input graph is an instance of $G(n,d/n)$ and $k$ is 
the numbers of colours.
\\ \vspace{-.7cm}\\

\noindent
{\bf Set up.}
Consider a sequence of graphs $G_0,\ldots,G_r$ such that every
$G_i$ is a subgraph of $G_{n,d/n}$. Moreover, $G_{r}$ is identical
to $G_{n,d/n}$, while $G_i$ is derived by deleting some edge of 
$G_{i+1}$. 

So as to get the graph $G_i$ from $G_{i+1}$ the only rule we follow 
is that we delete, arbitrarily, an edge that belongs to a sufficiently 
large cycle, i.e of length at least $(\log n)/(9\log d)$.
$G_0$ is the graph that comes up when there no are other such edges to delete.
Note only that $G_i$,  as a subgraph of $G(n,d/n)$, is somehow {\em random}.
\\ \vspace{-.8cm}\\

\noindent
{\bf Colouring.}
With probability $1-n^{-\Omega(1)}$, the sequence of subgraphs has the
property that $G_0$ is simple enough and we can $k$-colour it randomly 
in polynomial time by using some known algorithm.  In that case the algorithm
takes a random colouring of $G_0$. Then, for $i=0$ to $r-1$ it does the following:
it takes the random colouring of $G_i$, it does a simple, i.e. polynomial time, 
processing of this colouring and gets a random colouring of $G_{i+1}$. 
The algorithm continues until $G_r$.\\ \vspace{-.7cm}\\

\noindent
Let $G$ and $G'$ be two consecutive terms in the sequence of
graphs, above. Assume that $G$ is derived by deleting the edge
$\{v,u\}$ from $G'$. The critical question is the following one: Given
$X$, a random $k$-colouring of $G$, how can someone  use it to get
efficiently $X'$, a random $k$-colouring of $G'$. 
A moment's reflection makes it clear that if $X$ has the additional
property that $X(v)\neq X(u)$, then  
$X$ is distributed u.a.r. among the $k$-colourings of $G'$. In this case
we can simply set $X'=X$. Unfortunately, this cannot always be the
case and the random colouring algorithm we propose somehow deals
with situations as the one where $X(v)=X(u)$.

\begin{definition}[Good \& Bad colourings]
Let $\sigma$ be a proper $k$-colouring of $G$. We call $\sigma$ a
{\em bad} colouring of $G$ if $\sigma_v=\sigma_u$. Otherwise, we call
$\sigma$ a {\em good} colouring of $G$.
\end{definition}

\noindent
It turns out that the basic algorithmic challenge here is captured in 
the following problem.

\begin{problem}\label{prblm:StepProblem}
Given a {\em bad} random colouring of $G$, turn it to a {\em good} 
random colouring, in polynomial time.
\end{problem}

\noindent
Let us give an intuitive description of our algorithm for the above
problem. First remark the following: Consider $\sigma$, some  $k$-colouring
of $G$, and some $q\in[k]$ such that $\sigma_v\neq q$. It is easy to see 
that $\sigma$ specifies a {\em connected} subgraph of $G$ which includes
$v$ while every vertex in this  subgraph is assigned colouring either
$q$ or $\sigma_v$. The {\em maximal induced subgraph} of this kind
is called ``disagreement graph''\footnote{For a more formal definition
of ``disagreement graph'' see in Section \ref{sec:Problem1}.}.
Figure \ref{fig:G0} shows a $3$-colouring. The fat lines indicate the
disagreement  graph specified by using the colour ``g''.

It is direct to show that the disagreement graph that is specified
by the colouring $\sigma$ and the colour $q$ is always a connected,
bipartite graph whose parts are coloured $\sigma_v$ and $q$, respectively.

\begin{figure}
\begin{minipage}{0.5\textwidth}
	\centering
		\includegraphics[height=3.4cm]{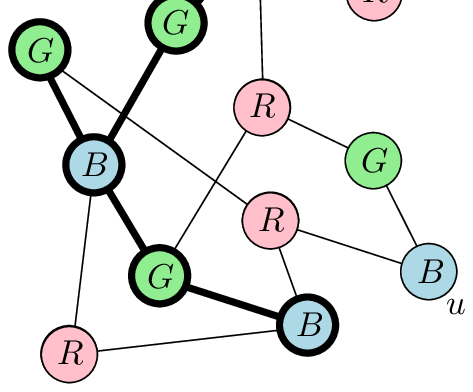}
		\caption{``Disagreement graph''.}
	\label{fig:G0}
\end{minipage}
\begin{minipage}{0.5\textwidth}
	\centering
		\includegraphics[height=3.4cm]{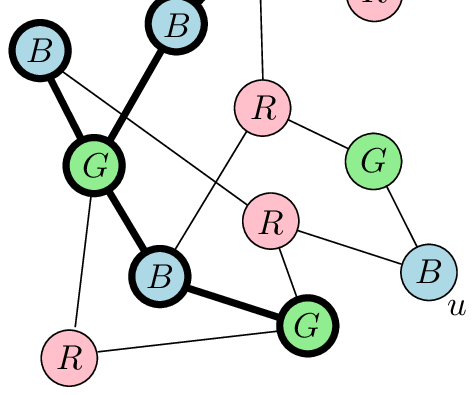}
		\caption{``g-switching''.}
	\label{fig:G1}
\end{minipage}	
\end{figure}

\begin{definition}
Assume that $\sigma$, a $k$-colouring of $G$,  and $q\in [k]
\backslash\{\sigma_v\}$ define the disagreement graph $Q$. The 
$k$-colouring of $G$,  $\sigma'$  is called ``$q$-switching of
$\sigma$'' if it is derived from $\sigma$ by switching the 
colour assignments of the vertices of $G$ that correspond 
to the two parts of $Q$.
\end{definition}

\noindent
In Figure \ref{fig:G1} we present the ``$g$-switching'' of the colouring
in Figure \ref{fig:G0}. It is direct that for the colouring $\sigma$
and for some $q\in[k]\backslash\{\sigma_v\}$ there is a unique
$q$-switching of $\sigma$. Also, it straightforward to show that
the $q$-switching of any proper $k$-colouring of $G$ is a proper
colouring, as well \footnote{ E.g. see proof of Lemma \ref{lemma:isomorphism}}.

Generally the $q$-switching of a bad colouring is not always a
good. However, given some technical conditions
which hold with probability 
$1-n^{-\Omega(1)}$ over the choices of $G$, we show the following,
{\em non-trivial}, statement
\begin{center}
\begin{tabular}{@{}p{0.9\textwidth}@{}}
The distribution of the $q$-switching of $Z$, a bad random 
$k$-colouring of $G$, is very close to the distribution of 
the good random $k$-colourings of $G$, when the colour $q$
is chosen uniformly at random from $[k]\backslash\{Z(v)\}$
and $k$ is sufficiently large.
\end{tabular}
\end{center}
The above fact suggests that we can have the the following {\em approximation algorithm} 
for Problem \ref{prblm:StepProblem} when  $G$ is a ``typical'' instance:  
Let $X$ be a random colouring $G$. If $X$ is {\em good}, then set
$X'=X$. If $X$ is a {\em bad}, then choose at random some $q\in 
[k]\backslash\{X(v)\}$ and set $X'$ to be equal to the $q$-switching 
of $X$. \\ \vspace{-.7cm}\\

\noindent
{\bf Remark.}
The algorithm in the previous paragraph is exactly the one we refer
in the introduction as {\tt STEP}.
\\ \vspace{-.7cm}\\

\noindent
Returning to the approximate random colouring algorithm,  we can 
build upon {\tt STEP} as follows. First, colour randomly 
$G_0$ with some known algorithm. Then, for $i=0$ to $r-1$ do the 
following: If the colouring of $G_i$ is {\em good}, then consider 
it as the colouring for $G_{i+1}$. Otherwise, choose appropriately 
a random colour $q$ and set as a colouring for the graph $G_{i+1}$
the $q$-switching of the colouring of $G_i$.

The above is a concise description of our approximate random colouring
algorithm. Clearly it is efficient and accurate only for typical instances 
of the input graph $G(n,d/n)$, i.e. it has the properties described
by Theorem \ref{thrm:GnpAccuracy} and Theorem \ref{thrm:GnpTimeCmplxt}.
\\ \vspace{-.7cm}\\

\remove{
\noindent
{\bf Remark.}
So as to get the graph $G_i$ from $G_{i+1}$ the only rule we follow 
is that we delete an edge that belongs to a sufficiently large cycle,
i.e of length at least $(\log n)/(9\log d)$. 
We show that in a $G_{n,d/n}$ after having deleted all 
edges that belong to these  large cycles we end up with 
a graph\footnote{Which contains only  the edges which belongs 
to small cycles} which is easy to $k$-colour randomly,
with probability $1-n^{-\Omega(1)}$.
}

\subsection{Some further remarks}\label{sec:pathological}
To get a better intuition about the algorithm {\tt STEP} we focus on
a case where things go wrong, i.e. consider the following.
Let $\sigma$ be  bad colouring of $G$, i.e. $\sigma_v=\sigma_u$.
It is possible that the disagreement graph specified by $\sigma$
and some colour $q$ to be so large that it contains both $v$ and
$u$. In this case the $q$-switching of $\sigma$ is a bad colouring.
Clearly in this case {\tt STEP} fails to generate a good colouring
of $G$.
Moreover, it is possible to have good colourings
of $G$ that cannot  be generated by applying the algorithm {\tt STEP}
to any bad colouring of $G$. Such colourings constitute 
{\em pathological} cases for the algorithm. These pathological
cases do not cause big problem as long as they occur rarely, i.e.
the fraction of colourings of $G$ that causes such situation is 
sufficiently small. The occurrences of {pathological} cases are
rare when $k$ is large and $v$, $u$ are far apart.

\section{Problem \ref{prblm:StepProblem} and $\alpha$-isomorphism}
{\tt STEP} uses the idea of $q$-switching so as to achieve a
certain kind of mapping between bad and good colourings. 
{\em Ideally} this mapping should have the property that, for
a bad random colouring of $G$ on the input, the image should be 
a good random colouring of $G$.
Unfortunately the $q$-switchings (as implemented by  {\tt STEP})
do not have this property but somehow they {\em approximate} such
mapping. We introduce few notions which capture the essence
of these ideas.
For the the following definitions in this section consider a fixed 
graph $G$ and let $\Omega$ be the sets of its proper $k$-colourings
\footnote{Take $k$ sufficiently large that $\Omega$ is non-empty.}.

\begin{definition}[Isomorphism]
We let  $\Omega_1, \Omega_2\subseteq \Omega$. We say that $\Omega_1$ 
is {\em isomorphic} to $\Omega_2$ if and only if there is a bijection 
$T:\Omega_1\to \Omega_2$.
\end{definition}

\noindent
The basic property of isomorphism we need here is contained in
the following corollary.

\begin{corollary}\label{corollary:IsoSmpl}
Assume that we have two isomorphic sets $\Omega_1$ and $\Omega_2$
and let $T$ be a bijection between these two sets. Then, 
given $X_1$, a random member of $\Omega_1$, the distribution
of $T(X_1)$ is the uniform over $\Omega_2$.
\end{corollary}

\noindent
The proof of Corollary \ref{corollary:IsoSmpl} appears in Section
\ref{sec:corollary:IsoSmpl}.
The previous definition of isomorphism is standard and generally
it expresses a notion of ``si\-mi\-la\-ri\-ty''.
We will need to get a bit further from this, i.e. we introduce
a more  general notion of ``similarity'' between  sets of colourings 
which we call $\alpha$-isomorphism.

\begin{definition}[$\alpha$-isomorphism]
We let $\Omega_1, \Omega_2\subseteq \Omega$ and $\alpha\in [0,1]$.
We say that $\Omega_1$ is {\em $\alpha$-isomorphic} to 
$\Omega_2$ if there are sets $\Omega'_1\subseteq \Omega_1$ and 
$\Omega'_2\subseteq \Omega_2$ such that \vspace{-.06cm}
\begin{itemize}
\item $|\Omega'_i| \geq (1-\alpha)|\Omega_i|$, for $i=1,2$. \vspace{-.06cm}
\item $\Omega'_1$ and $\Omega'_2$ are isomorphic.\vspace{-.06cm}
\end{itemize}
We call $(\Omega'_1, \Omega'_2)$ as the {\em isomorphic pair of
$\Omega_1$ and $\Omega_2$}.
\end{definition}

\noindent
Thus, rather than asking for the whole sets $\Omega_1$ and $\Omega_2$
to be isomorphic, $\alpha$-isomorphicity requires only sufficiently
large subsets from each of $\Omega_1$ and $\Omega_2$ to be isomorphic.
The notion of $\alpha$-function, that follows, is for $\alpha$-isomorphism 
the analogous of the bijection for isomorphism.

\begin{definition}[$\alpha$-function]\label{def:AFunction}
For some $\alpha \in [0,1]$, let $\Omega_1, \Omega_2\subseteq \Omega$
be two sets such that $\Omega_1$ is {\em $\alpha$-isomorphic} to $\Omega_2$
with isomorphic pair $(\Omega'_1, \Omega'_2)$. Let  $h:\Omega'_1
\to \Omega'_2$ be a bijection. Then the function $H: \Omega_1\to
[k]^V$ is called $\alpha$-function {\em if and only if} $\forall \sigma \in \Omega'_1$ 
it holds that $H(\sigma)=h(\sigma)$.
\end{definition}

\noindent
Note that we  can be a bit loose on the definition of an $\alpha$-function
when the input $\sigma$ does not belong to $\Omega'_1$, i.e. we
allow the $\alpha$-function to take any value in $[k]^V$. 
Showing that two sets $\Omega_1$ and $\Omega_2$
are $\alpha$-isomorphic reduces to providing a function which has
the properties stated in Definition \ref{def:AFunction}.

Typically, we are given two sets of $k$-colourings of $G$, e.g.
$\Omega_1$ and $\Omega_2$, and we will be asked to {\em devise} an 
$\alpha$-function which then suggests that these  $\Omega_1$ 
and $\Omega_2$ are $\alpha$-isomorphic.
The challenge is to devise an $\alpha$-function $H$ which complies 
to the following {\em efficiency rules}: First, given some $\sigma 
\in \Omega_1$ we want $H(\sigma)$ to have as few different colour 
assignments from $\sigma$ as possible, while the vertices with the
different colour  assignments should be as close to each othere as
possible.
Second, the smaller $\alpha$ and $k$ are, the better.
The $q$-switchings we introduced in the previous section are 
examples  of $\alpha$-functions between certain sets of 
colourings of $G$.
The next lemma states the most important property of $\alpha$-isomorphism
and somehow it generalizes Corollary \ref{corollary:IsoSmpl}.

\begin{lemma}\label{lemma:BijectionTVD}
Assume that the set $\Omega_1$ is $\alpha$-isomorphic to $\Omega_2$,
and $H:\Omega_1\to[k]^V$ is an $\alpha$-function. Let $z$ be a random
variable distributed  uniformly over $\Omega_1$ and let $z'=H(z)$. 
Denote by $\nu$ the uniform distribution over $\Omega_2$ and $\nu'$
the distribution of  $z'$. It holds that
\begin{displaymath}
||\nu-\nu' ||\leq \alpha.
\end{displaymath}
\end{lemma}
\inappendix{
The proof of Lemma \ref{lemma:BijectionTVD} appears in the appendix, 
Section \ref{sec:lemma:BijectionTVD}.
}

\subsection{Dealing with Problem \ref{prblm:StepProblem}}\label{sec:Problem1}
In this section we focus on {\tt STEP}. 
For clarity reasons we describe the algorithm by assuming that the graph 
$G$ in Problem \ref{prblm:StepProblem} is some general fixed graph. 
A basic part of the presentation involves relating the accuracy of 
{\tt STEP} to $\alpha$-isomorphism between certain sets of $k$-colourings
of $G$.

Let us introduce some notation. Let $\Omega$ denote the set of
$k$-colourings of $G$ and for $c,q\in[k]$ we let $\Omega(c,q)
\subseteq\Omega$ denote all the $k$-colourings of $G$ that assign
$v$ and $u$  the colours $c,q$, respectively.
We define formally a disagreement graph as follows:

\begin{definition}[Disagreement graph]\label{def:DiagreementGraph}
For $\sigma\in \Omega$ and some $q\in [k]\backslash\{\sigma_v\}$
we let the disagreement graph $Q_{\sigma_v,q}=(V', E')$ be the
maximal induced subgraph of $G$ such that
\begin{displaymath}
V'=\left \{x \in V \left | 
\begin{array}{l}
\exists \textrm{ path  } w_0,\ldots,w_t, \textrm{ in $G$ such that:  } \\
w_0=v, w_t=x, \sigma(w_j)\in \{\sigma_v, q\}
\end{array}
\right .
\right \}.
\end{displaymath}
\end{definition}

\noindent
It is important to remember that the disagreement graph is always connected, bipartite 
and maximal, i.e. for every $\sigma$ and $q$, $G$ has no vertex $y \notin Q_{\sigma_v,q}$
which has a neighbour in $V'$ and at the same time $\sigma_y\in \{\sigma_v,
q\}$. Furthermore, we define formally the $q$-switchings as a function $H:\Omega
\times[k]\to[k]^V$, i.e. $H(\sigma,q)$ returns the $q$-switching
of $\sigma$.  \\ \vspace{-.7cm} \\

\noindent
{\bf Function $H$}\\ \vspace{-.7cm} \\
\rule{\textwidth}{1pt}\\
{\bf Input:} $X\in \Omega$ and $q\in [k]\backslash\{X(v)\}$ \\
\hspace*{.6cm} Set $c=X(v)$.\\ 
\hspace*{.6cm} Set $V_1=\{w\in Q_{X(v),q}| X(w)=c\}$.  \\
\vspace{-.4cm}\\
\hspace*{.6cm} Set $V_2=\{w\in Q_{X(v),q})| X(w)=q\}$.  \\ 
\vspace{-.4cm}\\
\hspace*{.6cm} $\forall w\in V_1$  set $X(w)=q$.\\
\hspace*{.6cm} $\forall w\in V_2$  set $X(w)=c$.\\
{\bf Output:} $X$\\ \vspace{-.7cm} \\
\rule{\textwidth}{1pt} \\ 

\noindent
We have reached to the point where we  have all the definitions
we need to describe the algorithm {\tt STEP}.
\\ \vspace{-.7cm}\\

\noindent
{\bf STEP} \\ \vspace{-.7cm}\\
\rule{\textwidth}{1pt}
{\bf Input:} $X\in \Omega$, and  $k$ \\
\hspace*{.6cm} If $X$ is a  {\em good} colouring of $G$, then set  $Y=X$.\\
\hspace*{.6cm} If $X$ is a {\em bad}  colouring of $G$, then choose $q$
u.a.r. from $[k]\backslash\{X(v)\}$ and set $Y=H(X,q)$. \\ \vspace{-.5cm}\\
{\bf Output:} Y\\ \vspace{-.7cm}\\
\rule{\textwidth}{1pt} \\ \vspace{-.7cm}\\

\noindent
As far as the accuracy of {\tt STEP} is concerned we have to show that
if $X$ is a bad random $k$-colouring of $G$, then $H(X,q)$,
as calculated by the algorithm, is distributed sufficiently close to the desired di\-stri\-bu\-tion 
To this end $\alpha$-isomorphism comes into use. 

For any $c,q\in [k]$ we let $S(c,c)\subseteq \Omega(c,c)$
and $S(q,c)\subseteq \Omega(q,c)$ be defined as follows:
The set $S(c,c)$ contains every $\sigma\in \Omega(c,c)$ 
with the property that the disagreement graph $Q_{\sigma_v,q}$ 
does not contain the vertex $u$. Similarly, 
$S(q,c)$ contains every $\sigma\in \Omega(q,c)$ such that
the disagreement graph $Q_{\sigma_v,c}$ does not contain 
the vertex $u$.
We show that these two sets are isomorphic.
\begin{lemma}\label{lemma:isomorphism}
For any $c,q\in [k]$ with $c\neq q$, it holds that $S(c,c)$ and $S(q,c)$ are
isomorphic and the function $H(\cdot, q):S(c,c)\to S(q,c)$ is a
bijection.
\end{lemma}
\myshow{
The proof of Lemma \ref{lemma:isomorphism} appears in Section 
\ref{sec:lemma:isomorphism}.
}
\inappendix{
The proof of Lemma \ref{lemma:isomorphism} appears in the Section
\ref{sec:lemma:isomorphism}.
}
Based on the previous consideration, we provide a general relation 
between $\alpha$-isomorphism and the accuracy of {\tt STEP}. 
We make the following assumption.
\begin{assumption}\label{assm:isomorphism}
For some $\alpha\in[0,1]$ it holds that $\Omega(c,c)$ is $\alpha$-isomorphic
to $\Omega(q,c)$ with $\alpha$-function $H(\cdot, q)$, for any $c,q\in [k]$
such that $c\neq q$.
\end{assumption}

\noindent
Clearly, $(S(c,c),S(q,c))$ is the isomorphic pair of the $\alpha$-isomorphism between
$\Omega(c,c)$ and $\Omega(q,c)$.
Assumption \ref{assm:isomorphism} imposes an upper bound for the number of 
{\em pathological}  colourings \footnote{see also discussion in Section 
\ref{sec:pathological}} of $G$. It implies that, for any $c,q\in [k]$ all but 
an $\alpha$ fraction of the colourings in
$\Omega(q,c)$ do not have disagreement graph $Q_{q,c}$ which 
includes both $v$ and $u$. The same should hold for 
$\Omega(c,c)$ for disagreement graph $Q_{c,q}$.

\begin{theorem}\label{theorem:STEPAccuracy}
Let $\nu$ be the uniform distribution over the {\em good} $k$-colourings
of $G$. Let, also, $\nu'$ be the di\-stri\-bu\-tion of the output of {\tt STEP}
when the input colouring is distributed uniformly over the $k$-colourings
of $G$. Under Assumption \ref{assm:isomorphism} it holds that 
\begin{displaymath}
||\nu-\nu'||\leq \alpha,
\end{displaymath}
where $\alpha$ is defined in Assumption \ref{assm:isomorphism}.
\end{theorem}
The proof of Theorem \ref{theorem:STEPAccuracy} appears in the appendix,
Section \ref{sec:theorem:STEPAccuracy}. The impact of Assumption 
\ref{assm:isomorphism} to the accuracy of {\tt STEP} is apparent.

The value of $\alpha$ in Assumption \ref{assm:isomorphism} depends on
$G$, $k$ and the function $H$.
The natural way of considering that $\Omega(c,c)$ is $\alpha$-isomorphic
to $\Omega(c,q)$ is mainly as a consequence of  the $\alpha$-function 
$H(\cdot, q)$. That is, the two sets have this property because we have
devised a mapping, the $H(\cdot, q)$, which happens to be an 
$\alpha$-function between the two sets. 
Consequently, someone could device a ``better'' function, i.e. an 
$\alpha'$-function for  $\Omega(c,c)$ and $\Omega(c,q)$ such 
that either $\alpha'<\alpha$, or $\alpha'=\alpha$ but allowing 
smaller $k$, or both. 

Since the algorithm {\tt STEP} implements the $\alpha$-function,
the performance of the $\alpha$-function reflects the  performance 
of the algorithm itself. Clearly, 
the $\alpha$-function should be computable in polynomial time. 

\begin{lemma}\label{lemma:StepAc}
For a graph $G=(V,E)$ and some integer $k$, the time  complexity 
of computing the function 
$H(\cdot, q)$ is $O(|E|)$.
\end{lemma}
\inappendix{
The proof of Lemma \ref{lemma:StepAc} appears in Section \ref{sec:lemma:StepAc}.
}
\myshow{
\begin{proof}
Note that the time complexity of computing the value of $H(\sigma, q)$
is dominated by the time we need to reveal the disagreement graph 
$Q_{\sigma_v,q}$, for some $q\in [k]$. We show that we need $O(|E|)$ 
steps to reveal the disagreement graph $Q_{\sigma_v,q}$.

We can reveal the graph $Q_{\sigma_v,q}$ in steps $j=0,\ldots, |E|$, 
where $E$ is the set of edges of $G$.
At step $0$ the disagreement graph $Q_{\sigma_v,q}(0)$ contains only
the vertex $v$. Given the graph $Q_{\sigma_v,q}(j)$ we construct
$Q_{\sigma_v,q}(j+1)$ as follows: Pick some edge which is incident
to a vertex in $Q_{\sigma_v,q}(j)$. If the other end of this edge 
is incident to a vertex outside $Q_{\sigma_v,q}(j)$ that is coloured
either $\sigma_v$ or $q$ then we get $Q_{\sigma_v,q}(j+1)$ by inserting
this edge and the vertex into $Q_{\sigma_v,q}(j)$. Otherwise 
$Q_{\sigma_v,q}(j+1)$ is the same as $Q_{\sigma_v,q}(j)$. We never
pick the same edge twice. 

It is direct to show that in the above procedure it holds that
$Q_{\sigma_v,q}=Q_{\sigma_v,q}(|E|)$. Thus. the time complexity
of a $q$-switching of a given colouring of $G$ is $O(|E|)$.
\end{proof}

\noindent 
Also the following corollary is straightforward.
}

\randbits{
\begin{corollary}\label{cor:NoBitsStep}
The algorithm {\tt STEP} requires at most $\log_2 k$ random bits.
\end{corollary}
}

\section{From the algorithm Step to Random Colouring.}\label{sec:RCAlgorithm}
Here, we give a general presentation of the approximate random
colouring algorithm, which builds upon {\tt STEP}. We also study
properties of the algorithm like time complexity and accuracy.
In particular, we study the accuracy of the algorithm under 
general assumptions about $\alpha$-isomorphism, as we did in 
Section \ref{sec:Problem1} for {\tt STEP}. As in the previous cases,
the input graph $G$ is considered to be fixed.

First, we extend the notation of the previous section to fit here.
For input graph $G$ the algorithm considers the sequence of subgraphs
$G_0, G_1, \ldots, G_r$. $G_{i}$ is derived by deleting from $G_{i+1}$
an edge which we call $\{v_i, u_i\}$.
Let $\Omega_i$ be the set of  $k$-colourings of $G_i$. For any 
$c,q \in [k]$ we let $\Omega_i(c,q)$ be the set of colourings of $G_i$ 
which assign the colours $c$ and $q$ to the vertices $v_i$ and $u_i$, 
respectively.

We proceed by describing the full algorithm in pseudocode. The 
variable $Y_i$, below, denotes the $k$-colouring that the algorithm
assigns to the graph $G_i$.
\\ \vspace{-.7cm}\\

\noindent
{\bf Random Colouring Algorithm}\\ \vspace{-.7cm} \\
\rule{\textwidth}{1pt} \\ 
{\bf Input:} $G$, $k$.\\ 
Compute $G_0,G_1\ldots, G_r$.\\
Compute $Y_0$.  $\qquad \qquad\qquad /*$ {\em Get a random 
$k$-colouring of $G_0$}.$*/$\\
For $0\leq i \leq r-1$ do \\
\hspace*{.6cm} Set $Y_{i+1}$ the output of {\tt STEP} with input
$Y_i$.\\
{\bf Output:} $Y_{r}$.\\ \vspace{-.7cm} \\
\rule{\textwidth}{1pt} \\ \vspace{-.7cm} \\

\noindent
In the second line the algorithm computes the sequence of subgraphs
and in the third it colours randomly $G_0$. 
A detailed description of how can someone construct the sequence
of subgraphs and colour randomly $G_0$ is a graph specific problem.
For the case where the input graph is an instance of $G(n,d/n)$,
we give a detailed treatment of in the proof of 
Theorem \ref{thrm:GnpAccuracy} and Theorem \ref{thrm:GnpTimeCmplxt}
\footnote{See Section \ref{sec:AppGnp}.},
However, using Lemma \ref{lemma:StepAc} it is direct to get the
following theorem.

\begin{theorem}\label{thrm:TimeCmplxt}
Under the condition that $G_0$ can be $k$-coloured randomly in polynomial time, 
the random colouring algorithm has polynomial time complexity.
\end{theorem}

\noindent
The next issue we have to investigate is the accuracy of the
algorithm. As in Section \ref{sec:Problem1} we relate the accuracy
of the random $k$-colouring algorithm with $\alpha$-isomorphism by
using the following assumption.

\begin{assumption}\label{assm:SeqIsomorphism}
For $i=0,\ldots,r-1$ and some $\alpha\in [0,1]$ it hold that $\Omega_i(c,c)$
is $\alpha$-isomorphic to $\Omega_i(q,c)$  and $H(\cdot,q)$ is
a $\alpha$-function, for any $c,q\in [k]$ such that $c\neq q$.
\end{assumption}

\noindent
The $\alpha$-function $H$ is the same as the one defined in
Section \ref{sec:Problem1}. Let $(S_i(c,c),S_i(q,c))$ be the 
isomorphic pair of the $\alpha$-isomorphism between $\Omega_i(c,c)$
and $\Omega_i(q,c)$. The sets $S_i(c,c)$ and $S_i(q,c)$ are defined in 
the same manner as $S(c,c)$ and $S(q,c)$, in Section \ref{sec:Problem1}.
From Assumption \ref{assm:SeqIsomorphism}, we get the following theorem.
\begin{theorem}\label{thrm:Accuracy}
Let $\mu$ be the uniform distribution over the $k$-colourings of
the input graph $G$. Let, also, $\mu'$ be the distribution of the
colourings that is returned by the random colouring algorithm.
Under Assumption \ref{assm:SeqIsomorphism}, it holds that
\begin{displaymath}
||\mu-\hat{\mu}||\leq r\cdot \alpha,
\end{displaymath}
where $r$ is the maximum index in the sequence $G_0,G_1,\ldots, G_r$.
\end{theorem}
\inappendix{
The proof of Theorem \ref{thrm:Accuracy} appears in the appendix, 
Section \ref{sec:thrm:Accuracy}.
}

\inappendix{
\section{Proof sketch for Theorem \ref{thrm:GnpAccuracy}}
Due to space limitations, in the remaining pages we give a proof 
sketch our main result,   Theorem \ref{thrm:GnpAccuracy}.
That is, we consider the random colouring algorithm with input 
an instance of $G(n,d/n)$ and we let $k$ be the number of colours. 
From a technical perspective there are two issues to deal with. The
first is how do we construct the sequence of subgraphs. The second is 
to replace the rather general Assumption \ref{assm:SeqIsomorphism} 
about $\alpha$-isomorphism between colour sets with specific results 
for the graphs $G_0,G_1, \ldots, G_r$.

When the algorithm constructs the sequence of subgraphs it should 
take into consideration the previous remark that it is preferable
in the graph $G_i$ the vertices $v_i$ and $u_i$ to be at a sufficiently
large distance.
To see why we need this property we provide the following corollary, 
which  follows directly from previous definitions.
\begin{corollary}\label{sketch:cor:CombCond}
Consider some fixed graph $G_i$ and $c,q\in [k]$. The set $\Omega_i(c,c)$
is $\alpha$-isomorphic to $\Omega_i(q,c)$ with $\alpha$-function
$H(\cdot, q)$ if and only if the following holds:
Choose u.a.r. a colouring from $\Omega_i(c,c)$ and let $Q_{c,q}$
be the disagreement graph specified by this colouring and $q$.
It should hold that
\begin{equation}\label{sketch:eq:CombCondA}
a\geq Pr[v_i,u_i\in Q_{c,q}|G_i].
\end{equation}
Additionally, the analogous condition should hold for a random colouring
of $\Omega_i(q,c)$.
\end{corollary}
Since we are interested in the minimum possible value for $\alpha$,
we try to minimize the probability term in (\ref{sketch:eq:CombCondA}).
Clearly,  the greater the distance between $v_i$ and $u_i$ the
less probable is for $ Q_{c,q}$ to include them both and, 
consequently, 
the more accurate the random colouring algorithm gets. 
To this end we use the following lemma to construct the sequence of subgraphs. 
\begin{lemma}\label{sketch:lemma:seqsubprop}
With probability at least $1-n^{-2/3}$ we can have the sequence
$G_0,\ldots G_r$ satisfying the following two properties.\vspace{-.1cm}
\begin{enumerate}
\item $G_0$ consists only of isolated vertices and simple cycles, 
each of maximum length less than $\frac{\log n}{9\log d}$. \vspace{-.1cm}
\item In $G_i$, the graph distance between $v_i$ and $u_i$ is at least $\frac{\log n}{9\log d}$.\vspace{-.1cm}
\end{enumerate}
Additionally it holds that 
\begin{displaymath}
Pr[r\geq (1+n^{-1/3})dn/2]\leq \exp(-n^{1/4}).
\end{displaymath}
\end{lemma}

\noindent
In the rest of the analysis of the algorithm we assume that the sequence 
of subgraphs is such that the distance between $v_i$ and $u_i$
is at least $\gamma\log n$, where $\gamma=(9\log d)^{-1}$, 
for $i=0,\ldots, r-1$.
Since $G_0,\ldots G_{r}$ are subgraphs of $G_{n,d/n}$, somehow
they are random too. The reader should feel free to assume any 
{\em arbitrary} rule that generates $G_i$ from each instance of
$G(n,d/n)$. The only restriction we have is that of the distance
between $v_i$ and $u_i$. 
%
\remove{
\noindent
{\bf Remark.}
The reader should  note that when we are  dealing with a random colouring
of $G_i$, we have to deal with {\em two levels of randomness}.
The first level is the random graph $G_i$. Having fixed the graph instance, 
the second level is its random colouring.\\ \vspace{-.8cm}\\
}
%
Then, we use the following theorem. 

\begin{theorem}\label{sketch:thrm:eIsomGnp}
Take $k\geq(2+\epsilon)d$, where  $\epsilon>0$ and $d$ is a sufficiently
large fixed number. There is $\beta_i$ such that for any $\alpha\geq \beta_i$
it holds that $\Omega_i(c,c)$ is $\alpha$-isomorphic to $\Omega_i
(c,q)$ and $H(\cdot,q)$ is an $\alpha$-function while  
\begin{displaymath}
E[\beta_i]\leq C\cdot n^{-\left(1+\frac{\epsilon}{45\log d}\right)},
\end{displaymath}
for any $c,q\in [k]$, $C>0$ is fixed  and $i=0,\ldots,r-1$. 
\end{theorem}

\noindent
The expectation of the quantity $\beta_i$ is over the graph instances $G(n,d/n)$.
Taking $r_0=(1+n^{-1/3})dn/2$, Theorem \ref{thrm:Accuracy} implies that
\begin{displaymath}
E\left [||\mu-\hat{\mu}|| \right]\leq E \left[ \sum_{i=0}^r \beta_i \right]
\leq \displaystyle \sum_{i=0}^{r_0} E[\beta_i|r\leq r_0]+n^2 Pr[r\geq r_0].
\end{displaymath}
In the last inequality we use that $\beta_i\in [0,1]$.
It easy to see that $E[\beta_i|r\leq r_0]\leq Pr^{-1}[r\leq r_0]\cdot E[\beta_i]$.
Then, Theorem \ref{sketch:thrm:eIsomGnp} and Lemma \ref{sketch:lemma:seqsubprop}
suggest that there is fixed $C>0$ such that $E\left [ ||\mu-\hat{\mu}|| \right ] 
\leq C\cdot n^{-\frac{\epsilon}{45\log d}}.$
The theorem follows by applying the Markov inequality.

\subsection{Proof sketch for Theorem \ref{sketch:thrm:eIsomGnp}}
\label{sec:sketch:thrm:eIsomGnp}
Consider some fixed instance of $G_i$ and let $c,q\in [k]$ such that 
$c\neq q$. Choose u.a.r. a colouring from $\Omega_i(c,c)$ and let
$Q_{c,q}$ be the disagreement graph that is specified by the colouring
we chose and $q$. Similarly, choose u.a.r. from $\Omega_i(q,c)$ and 
let $Q_{q,c}$ be the disagreement graph specified by the chosen
colouring and $c$.
According to Corollary \ref{sketch:cor:CombCond}, $\Omega_i(c,c)$
is $\alpha$-isomoprhic to $\Omega(q,c)$ for any $\alpha\geq \beta_i$
such that
$\beta_i=Pr[v_i,u_i\in Q_{c,q}|G_i]+Pr[v_i,u_i\in Q_{q,c}|G_i].$
Taking the average over $G_i$ we have
\begin{equation}\label{sketch:eq:EXPCT}
E[\beta_i]=Pr[v_i,u_i\in Q_{c,q}]+Pr[v_i,u_i\in Q_{q,c}].
\end{equation}
We provide a  bound on the probability terms in (\ref{sketch:eq:EXPCT}) by using 
the following proposition.

\begin{proposition}\label{sketch:prop:DecayOfPaths}
Take $k\geq (2+\epsilon)d$, for fixed $\epsilon>0$.
Let $\sigma$ be a $k$-colouring of $G_i$ that is chosen u.a.r. among $\cup_{c'
\in [k]}\Omega_i(c,c')$. For some $q\in [k]\backslash\{c\}$ we let
the event $A_i=$``$v_i$ and $u_i \in Q_{\sigma_{v_i,q}}$''. 
There is a positive constant $C$ such that 
\begin{displaymath}
Pr[A_i]\leq C\cdot n^{-\left(1+\frac{\epsilon}{45\log d} \right)}
\qquad i=0,\ldots,r.
\end{displaymath}
\end{proposition}

\noindent
Note that in the above proposition we take a disagreement
graph of a colouring chosen u.a.r. among the colourings of an instance of $G_i$
that assign the vertex $v_i$ the colour $c$ (the colouring of $u_i$ is
``free''). We show that  choosing a u.a.r. a colouring from 
$\cup_{c'\in [k]}\Omega_i(c,c')$ the probability $p$ for this colouring 
to be in $\Omega_i(c,c)$ is {\em constant}, i.e. $p=\Theta(1)$. 
Then, the law of total probability suggests that $Pr[v_i,u_i\in Q_{c,q}]
\leq p\cdot Pr[A_i]$. We work similarly for $Pr[v_i,u_i\in Q_{q,c}]$. 
The theorem follows.

\subsection{Proof sketch for Proposition \ref{sketch:prop:DecayOfPaths}}
In the experiment in the statement of Proposition \ref{sketch:prop:DecayOfPaths},
we 	let $W_i(l)$ denote the number of paths in $Q_{c,q}$ that start at $v_i$ and
end at $u_i$ and have length $l$. By the Markov inequality we get that 
\begin{equation}\label{eq:sketch:first}
Pr[A_i]\leq \sum_{l=\gamma\log n}^{\infty} E\left[W_i(l)\right],
\end{equation}
where $\gamma=(9\log d)^{-1}$. Thus it remains to bound the expectation
on the r.h.s.

For a vertex $w$, we let $deg_i(w)$ be its degree in the graph $G_i$.
Consider the product measure ${\cal P}(G_i,k)$ such that each vertex 
$w \in G_i$ is {\em disagreeing} with probability $q_w=\frac{1}{k-deg_i(w)}$ 
and {\em non-disagreeing} with probability $1-q_w$. Also, the vertex
$v_i$ is disagreeing with probability 1. When $k\leq deg_i(w)$ we
set $q_w=1$. A {\em path of disagreement} in $G_i$ is any simple
path which has all its vertices disagreeing.

Let $\Gamma_i(l)$ denote the number of paths of disagreement between
$v_i$ and $u_i$ in $G_i$, in a configuration chosen according to
${\cal P}(G_i,k)$. Through a stochastic order relation we show that
for any $l$ it holds
\begin{equation}\label{eq:sketch:second}
E\left[W_i(l)\right]\leq E_{\cal P}\left[\Gamma_i(l)\right],
\end{equation}
where the rightmost expectation is w.r.t. both the measure ${\cal P}
(G_i,k)$ and $G_i$. Then taking $k\geq(2+\epsilon)d$ and sufficiently
large $d$ it holds that 
\begin{equation}\label{eq:sketch:third}
E_{\cal P}[\Gamma_i(l)]\leq \Theta(1)\cdot n^{l-1}\left( \frac{d}{n}\right)^l \cdot 
\left( \frac{1}{(1+\epsilon/5)d}\right)^l = \Theta(1)\frac{1}{n}(1+\epsilon/5)^{-l}.
\end{equation}
The coefficient $n^{l-1}$ comes from the fact that between $v_i$ and $v_i$
there are at most $n^{l-1}$ paths of length $l$, $(d/n)^l$ is an upper bound 
for the probability to have a specific path of length $l$ in $G_i$ 
and the final coefficient is related to the probability for a path of length 
$l$ to be a ``path of disagreement''. The proposition follows by combining
(\ref{eq:sketch:first}),(\ref{eq:sketch:second}) and (\ref{eq:sketch:third}).

To get a better picture of why there are not many paths of disagreement
when $k\geq(2+\epsilon)d$ consider $q_w$ the marginal distribution of
$w$ in $G_i$ to be disagreeing. 
For $k\geq (2+\epsilon)d$ it holds that 
\begin{displaymath}
q_w \leq \frac{1}{(1+\epsilon/2)d}+Pr[deg_i(w)>(1+\epsilon/2)d+1].
\end{displaymath}
Clearly, $deg_i(w)$ is dominated by ${\cal B}(n,d/n)$. Using Chernoff bounds
\footnote{Corollary 2.4 in \cite{janson}.} we can show that for any fixed 
$\epsilon$, the rightmost probability is smaller than $e^{c'd}$, for fixed $c'$.
Then, roughly speaking, we have the following situation:
The expected degree of $w$ is at most $d$. Also, $w$ is disagreeing with 
probability $q_w<1/d$, for sufficiently large $d$. Consequently, 
for every path of disagreement that enters $w$ the expected number
of paths that leave $w$ are $d\cdot q_w<1$.


}

\inappendix
{
\newpage
\renewcommand{\thepage}{A.\arabic{page}}
\renewcommand{\thesection}{\Alph{section}}
\setcounter{page}{1}
\setcounter{section}{0}

\noindent
{\LARGE \bf Appendix }
}

\section{Proof of Theorem \ref{thrm:GnpAccuracy} and Theorem \ref{thrm:GnpTimeCmplxt}}\label{sec:AppGnp}
In this section we use results from Section \ref{sec:RCAlgorithm}
to show Theorem \ref{thrm:GnpAccuracy} and Theorem \ref{thrm:GnpTimeCmplxt}
when the input of the random colouring algorithm is an instance of
$G(n,d/n)$. Essentially there are two issues to deal with, the
first is how do we construct the sequence of subgraphs, while the
second is to replace the, rather general, Assumptions \ref{assm:SeqIsomorphism} 
with more specific results for the colourings of the graphs $G_0,G_1,\ldots, G_r$.

It is easy to construct a sequence of subgraph so as to have $G_0$
randomly $k$-coloured in polynomial time (e.g. take it such that
$G_0$ is empty).  However, the actual
construction of the sequence of subgraphs is a bit more complicated
task. It has been remarked very early in this work that in the 
graph $G_i$ the vertices $v_i$ and $u_i$ are at a sufficiently 
large distance. To see why we need this property we provide the 
following corollary, which  follows directly from the definitions
in the previous sections.

\begin{corollary}\label{cor:CombCond}
Consider some fixed graph $G_i$ and $c,q\in [k]$. The set $\Omega_i(c,c)$
is $\alpha$-isomorphic to $\Omega_i(q,c)$ with $\alpha$-function
$H(\cdot, q)$ if and only if the following holds:
Choose u.a.r. a colouring from $\Omega_i(c,c)$ and let $Q_{c,q}$
be the disagreement graph specified by this colouring and $q$.
It should hold that
\begin{equation}\label{eq:CombCondA}
a\geq \max_{q\in[k]\backslash \{c\}}Pr[v_i,u_i\in Q_{c,q}|G_i]. 
\end{equation}
Additionally, the analogous condition should hold for a random 
colouring of $\Omega_i(q,c)$.
\end{corollary}

\noindent
Since we are interested in the minimum possible value for $\alpha$,
we see that the greater the distance between $v_i$ and $u_i$ the
less probable is for the disagreement graph to include them both. 
Thus, the greater the distance between $v_i$ and $u_i$ the more
accurate the random colouring algorithm gets. 
To this end we use
the following lemma to construct the sequence of subgraphs of $G_{n,d/n}$.

\begin{lemma}\label{lemma:seqsubprop}
With probability at least $1-n^{-2/3}$ we can have the sequence
$G_0,\ldots G_r$ satisfying the following two properties.
\begin{enumerate}
\item $G_0$ consists only of isolated vertices and simple cycles, 
each of maximum length less than $\frac{\log n}{9\log d}$.
\item In $G_i$, the graph distance between $v_i$ and $u_i$ is at least $\frac{\log n}{9\log d}$.
\end{enumerate}
Additionally it holds that 
\begin{displaymath}
Pr[r\geq (1+n^{-1/3})dn/2]\leq \exp(-n^{1/4}).
\end{displaymath}
\end{lemma}
The proof of Lemma \ref{lemma:seqsubprop} appears in Section \ref{sec:lemma:seqsubprop}.

In the analysis that follows,  we assume that the sequence of
subgraphs that is computed by the random colouring algorithm,
has the properties stated in Lemma \ref{lemma:seqsubprop}. 
Since $G_0,\ldots G_{r}$ are subgraphs of $G_{n,d/n}$, somehow
they are random too and they depend on $d$. The reader should feel
free to assume any, {\em arbitrary}, rule that generates $G_i$ from
each instance of $G(n,d/n)$. The only restriction we have is that
of the distance between $v_i$ and $u_i$.

\begin{theorem}\label{thrm:eIsomGnp}
Take $k\geq(2+\epsilon)d$, where  $\epsilon>0$ and $d$ is a sufficiently large
fixed number. There is $\beta_i$ such that for any $\alpha\geq \beta$
it holds that $\Omega_i(c,c)$ is $\alpha$-isomorphic to $\Omega_i
(c,q)$ and $H(\cdot,q)$ is an $\alpha$-function while
\begin{displaymath}
E[\beta_i]\leq \frac{(40+8\epsilon)k}{\epsilon}n^{-\left(1+\frac{\epsilon}{45\log d}\right)},
\end{displaymath}
for any $c,q\in [k]$ and $i=0,\ldots,r$.
\end{theorem}

\noindent
Since the graphs $G_i$ are random the corresponding sets $\Omega_i$ are random too.
The above expectation is taken w.r.t. the random graph $G_i$, for $i=0,\ldots, r-1$.
The proof of Theorem \ref{thrm:eIsomGnp} appears in Section \ref{sec:thrm:eIsomGnp}.
\\ \vspace{-.7cm} \\

\begin{theoremproof}{\ref{thrm:GnpAccuracy}}
Using Theorem \ref{thrm:Accuracy} and Theorem \ref{thrm:eIsomGnp}
we have that
\begin{displaymath}
E\left [ ||\mu-\hat{\mu}|| \right
 ]\leq E \left[ \sum_{i=0}^r \beta_i \right],
\end{displaymath}
where the expectation is taken over the instances of the input $G_{n,d/n}$.
Noting that  $\beta_i\in [0,1]$, we get 
\begin{displaymath}
E\left [ ||\mu-\hat{\mu}|| \right ] \leq \displaystyle \sum_{i=0}^{(1+n^{-1/3})dn/2} E[\beta_i|r\leq (1+n^{-1/3})dn/2]+
n^2 Pr[r\geq (1+n^{-1/3})dn/2].
\end{displaymath}
It is direct that  
$$
E[\beta_i|r\leq (1+n^{-1/3})dn/2]\leq Pr^{-1}[r\leq (1+n^{-1/3})dn/2]\cdot 
E[\beta_i]\leq \frac{3}{2} \frac{(40+8\epsilon)k}{\epsilon} n^{-(1+\frac{\epsilon}{45\log d})}
$$
in the final inequality we used Theorem \ref{thrm:eIsomGnp}.
Combining all the above with Lemma \ref{lemma:seqsubprop},  we get that
\begin{displaymath}
E\left [ ||\mu-\hat{\mu}|| \right ]\leq C\cdot n^{-\frac{\epsilon}{45\log d}}.
\end{displaymath}
for fixed $C>0$. The theorem follows by applying the Markov inequality.
\end{theoremproof}

%

\begin{theoremproof}{\ref{thrm:GnpTimeCmplxt}}
As we show in the proof of Lemma \ref{lemma:seqsubprop},
with probability at least $1-\exp(-n^{1/4})$, the number
of edges of $G_{n,d/n}$ is at most $(1+n^{-1/3})\frac{dn}{2}$.
From now on in the proof, assume that we are dealing with
a graph with $\Theta(n)$ edges. In this case it is direct 
that $r=\Theta(n)$, as well.

Since the number of edges is linear, we need $O(n)$ time to find
whether some edge belongs to a
small cycle, i.e. a cycle of length less than $\frac{\log n}{9\log d}$, 
or not.
This can be done by exploring the structure of the
$\frac{\log n}{9\log d}$-neighbourhood around this edge.
Thus, the algorithm requires $O(n^2)$ time to create 
the sequence of subgraph.

Also, it is clear that we need $O(n)$ time to implement one 
switching of a colouring. For more details on how this can be 
done see in the proof of Lemma \ref{lemma:StepAc}. Since $r=O(n)$, 
we need $O(n^2)$ time for the all colour switchings in the algorithm.

As far as the random colouring of $G_0$ is regarded we note
the following: 
Using Dynamic Programming we can compute exactly the number of
list colourings of a tree $T$. In the list colouring problem every
vertex $v \in T$ has a set $List(v)$ of valid colours, where
$List(v) \subseteq [k]$ and $v$ only receives a colour in $List(v)$. 
For a tree on $l$ vertices, using dynamic programming we can 
compute the exact number of list colourings in time $lk$.
For a unicyclic component, i.e. a tree with an extra edge,  we 
can consider all the $k^{2}$ colourings of the endpoints of the
extra edges and for each of these colourings recurse on the 
remaining tree. 
Thus, it is direct to show that we can have a random $k$- colouring
of $G_0$ in time $O(n)$.

The theorem follows by noting that the construction of the
sequence of subgraphs with the desired properties fails with probability
at most $n^{-2/3}$.
\end{theoremproof}

\section{Proof of Theorem \ref{thrm:eIsomGnp}}\label{sec:thrm:eIsomGnp}
So as to prove Theorem \ref{thrm:eIsomGnp} we use the following proposition.

\begin{proposition}\label{prop:DecayOfPaths}
Take $k\geq (2+\epsilon)d$, where $\epsilon>0$ and $d$ is a sufficiently 
large number.
Let $\sigma$ be a $k$-colouring of $G_i$ that is chosen u.a.r. among $\cup_{c'
\in [k]}\Omega_i(c,c')$. For some $q\in [k]\backslash\{c\}$ we let
the event $A_i=$``$v_i$ and $u_i \in Q_{\sigma_{v_i,q}}$''. 
It holds that 
\begin{displaymath}
Pr[A_i]\leq \frac{10+2\epsilon}{\epsilon}n^{-\left(1+\frac{\epsilon}{45\log d} \right)}
\qquad i=0,\ldots,r.
\end{displaymath}
\end{proposition}

\noindent
The reader should remark that since the graph $G_i$ is random,
for the probability term $Pr[A_i]$ in the proposition it holds
that
\begin{displaymath}
Pr[A_i]=E[Pr[A_i|G_i]],
\end{displaymath}
where the expectation w.r.t. $G_i$. The proof of Proposition \ref{prop:DecayOfPaths} appears in Section
\ref{sec:prop:DecayOfPaths}\\ \vspace{-.7cm}\\

\begin{theoremproof}{\ref{thrm:eIsomGnp}}
Consider, first, a fixed sequence $G_{i}$, for $i=0,1\ldots,r$.
Assume that we choose a $k$-colouring u.a.r. among $\Omega_i(c,c)$
and let $Q_{c,q}$ be the disagreement graph specified by the chosen
$k$-colouring and $q$. Let the event $B_i=$``$v_i,u_i\in Q_{c,q}$''
in the above experiment.

Similarly, assume that we choose u.a.r. a $k$-colouring from $\Omega_i(q,c)$
and let $Q_{q,c}$ be the disagreement graph  specified by the chosen
$k$-colouring and $q$. Let the event $C_i=$``$v_i,u_i\in Q_{c,q}$''
in this experiment.

We let $\beta_i=\max\{Pr[B_i|\Omega_{i}(c,c)],Pr[C_i|\Omega_{i}(q,c)]\}$.
Corollary \ref{cor:CombCond} implies that for any $\alpha\geq \beta_i$
it holds that the set $\Omega_i(c,c)$ is $\alpha$-isomorphic to
$\Omega_i(q,c)$ with $\alpha$-function $H(\cdot, q)$. 
Also, it is straightforward that 
$$E[\beta_i]\leq Pr[B_i]+Pr[C_i].$$
The above expectation is taken w.r.t. to the instances $G_i$.

Assume that we choose u.a.r. a member of a fixed instance of $\cup_{c'\in[k]} 
\Omega_i(c,c')$  and we denote with $E_i$ the event that the 
chosen colouring belongs to $\Omega_i(c,c)$. Also let
\begin{displaymath}
p=\sum_{G}Pr[E_i|G]\cdot {\cal D}_i[G],
\end{displaymath}
where, for a fixed graph $G$, $Pr[E_i|G]$ is equal to the
probability to have the event $E_i$ when the sets of $k$-colourings
are specified by the graph $G$. ${\cal D}_i(G)$ is equal to
the  probability that an instance of $G_i$ is the graph $G$. 
Applying the law of total probability we get that 
\begin{displaymath}
\begin{array}{lcl}
Pr[A_i]&=&Pr[A_i|E_i]Pr[E_i]+Pr[A_i|E^c_i]Pr[E^c_i]
\\ \vspace{-.3cm}\\
&\geq& Pr[A_i|E_i]Pr[E_i]=Pr[B_i]\cdot p.
\end{array}
\end{displaymath}
Thus, it holds that
\begin{displaymath}
Pr[B_i] \leq p^{-1} Pr[A_i].
\end{displaymath}
Since we have the value of $Pr[A_i]$ from Proposition \ref{prop:DecayOfPaths},
we only need to compute a lower bound for the probability $p$. 
For a fixed graph $G$, let $\mu_{G}$ denote the Gibbs distribution
of the $k$-colourings of $G$. Also let $\mu_{i}$ be defined
as follows:
\begin{displaymath}
\mu_i(\sigma)=\sum_{G}\mu_{G}(\sigma){\cal D}_i(G) \qquad \forall \sigma\in[k]^V.
\end{displaymath}
We use the following claim to compute bounds for $p$.
\begin{claim}\label{claim:spatial}
Taking $k\geq(2+\epsilon)d$, where $\epsilon>0$ is fixed and $d$
is a sufficiently large number, it holds that 
\begin{displaymath}
\max_{\sigma\in \Omega_i}||\mu_i(\cdot|\sigma_v)-\mu_i(\cdot)||_u\leq n^{-1}
\qquad i=0,\ldots, r.
\end{displaymath}
\end{claim}
It is easy to show that under $\mu_i$ the marginal distribution
of the colour assignment of the vertex $u$ is the uniform over
the set $[k]$. 
The above claim suggests that for $k\geq(2+\epsilon)d$ it holds that 
$$
\left |p-\frac{1}{k}\right |\leq n^{-1}.
$$
Thus we get that 
\begin{displaymath}
Pr[B_i]\leq \displaystyle  2k Pr[A_i]\leq \frac{(20+4\epsilon)k}{\epsilon}n^{-\left(1+\frac{\epsilon}{45\log d}\right)}.
\end{displaymath}
Using the same arguments we, also, get that
\begin{displaymath}
Pr[C_i]\leq \frac{(20+4\epsilon)k}{\epsilon}n^{-\left(1+\frac{\epsilon}{45\log d}\right)}.
\end{displaymath}
The theorem follows.
\end{theoremproof}

\begin{claimproof}{\ref{claim:spatial}}
First assume that we have a fixed $G_i$, i.e. the set of colourings is fixed.
Let $X_i$ be distributed uniformly over $\cup_{c'\in[k]}\Omega_i
(c,c')$ and let $Z_i$ be distributed uniformly over $\Omega_i$.
We couple these two variables . The coupling is done as follows. 
Choose u.a.r. a colour from $[k]$, and set $Z_i(v_i)$ equal to this
colour, e.g. let $Z_i(v_i)=q$. We have two cases.

If $q=c$, then we can have an identical coupling between $Z_{i}$ and $X_i$.
Otherwise, i.e. if $Z_{i}(v)\neq X_i(v_i)$, we can set $Z_{i}=H(X_i,q)$. 

\begin{claim}\label{claim:OldIsomorph}
In the later case, i.e. when $Z_i=H(X_i,q)$, $Z_i$ is distributed 
uniformly over the colouring of $G_i$ that assign the vertex $v_i$
the colour $q$.
\end{claim}
The proof of the claim appear after the end of this proof.

Thus, in the case where  $Z_{i}(v)\neq X_i(v)$ and we set $Z_{i}=H(X_i,q)$ 
it is direct to see that $Z_i(u_i)\neq X_i(u_i)$ if and only if the event 
$A_i$ (as defined in the statement of Proposition \ref{prop:DecayOfPaths}) 
holds. Thus we get that 
\begin{displaymath}
Pr[X_i(u_i)\neq Z_i(u_i)|G_i]\leq Pr[A_i|G_i]
\end{displaymath}
From the above relation and Proposition
\ref{prop:DecayOfPaths} we get that
\begin{displaymath}
Pr[X_i(u_i)\neq Z_i(u_i)]\leq Pr[A_i]\leq n^{-1},
\end{displaymath}
The claim follows by using the Coupling Lemma.
\end{claimproof}

\begin{claimproof}{\ref{claim:OldIsomorph}}
We remind the reader that $X_i$ is distributed uniformly 
at random among the $k$-colourings of $G_i$ that assign 
the vertex $v_i$ the colour $c$. 
It suffice to show that the sets $\Omega_c=\cup_{c' \in [k]}
\Omega_i(c,c')$ and $\Omega_q=\cup_{c'\in[k]}\Omega_i(q,c')$ 
are isomorphic with bijection $H(\cdot,q):\Omega_c\to \Omega_q$. 
The arguments we need to show this are the same as those
we use in the proof of Lemma \ref{lemma:isomorphism}.

I.e. first we need to show that for any $\sigma\in \Omega_c$ it
holds that $H(\sigma,q)$ is a proper colouring of $G_i$. Clearly
this holds 
(see the first two paragraphs of the proof of Lemma \ref{lemma:isomorphism}
in section \ref{sec:lemma:isomorphism}.)
Second we need to show that the mapping $H(\cdot, q):\Omega_c\to \Omega_q$
is  {\em surjective}, i.e. for any $\sigma\in \Omega_q$ there is a 
$\sigma'\in \Omega_c$ such that $\sigma=H(\sigma',q)$. 
It is direct to see that such $\sigma'$ exists, moreover, it holds 
that $\sigma'=H(\sigma,c)$. 
Finally, we need to show that $H(\cdot, q)$ is {\em one-to-one}, 
i.e. there are no two $\sigma_1,\sigma_2\in \Omega_c$ such that
$H(\sigma_1,q)=H(\sigma_2,q)$. Using arguments similar to for
the  {\em surjective} case  it is direct to see that there cannot
be such a pair of colourings. The claim follows. 
\end{claimproof}

\subsection{Proof of Proposition \ref{prop:DecayOfPaths}}
\label{sec:prop:DecayOfPaths}

Consider the probability distribution ${\cal L}(G_i,k)$ (or ${\cal L}_{G_i,k}$)
induced by the following experiment. We have a graph $G_i$ and we
choose a $k$-colouring $\sigma$ u.a.r. from $\cup_{c'\in [k]}\Omega_i(c,c')$,
where $c\in [k]$. Choose u.a.r. a colour from $[k]\backslash\{c\}$, let
$q$ be that colour. Create the  graph of disagreement $Q_{c,q}$. 
If $w\in Q_{c,q}$, then $w$ is ``disagreeing'' otherwise it is  
``non-disagreeing''. By definition $v_i$ is always in the disagreement graph.

For a vertex $w$, we denote with $deg_i(w)$ its degree in the graph $G_i$.
Consider, also, the product measure ${\cal P}(G_i,k)$ such that
each vertex $w \in G_i$ is {\em disagreeing} with probability
$q_w=\frac{1}{k-deg_i(w)}$ and {\em non-disagreeing} with probability
$1-q_w$. Also, the vertex $v_i$ is disagreeing with probability 1. 
When $k\leq deg_i(w)$ we set $q_w=1$.

A {\em path of disagreement} in $G_i$ is any simple path which has
all its vertices disagreeing. The measure ${\cal P}(G_i,k)$ will
turn out to be very useful because it dominates ${\cal L}(G_i,k)$
in the following sense.

\begin{lemma}\label{lemma:StochOrderSimple}
Let $M=x_{1}, x_2, \ldots, x_l$ be a path in $G_i$ such that $v_i=x_1$.
Let the event $E=$''$M$ is a path of disagreement''. It holds that
\begin{displaymath}
{\cal L}_{G_i,k}[E]\leq {\cal P}_{G_i,k}[E].
\end{displaymath}
\end{lemma}
\begin{proof}
Let the event $E_i=$``$x_i$ is disagreeing'', for $i\leq l$, obviously
$E=\bigcap_{j=1}^lE_j$. It is direct that 
\begin{displaymath}
{\cal L}_{G_i,k}[E]={\cal L}[E_1]\prod_{j=2}^l {\cal L}_{G_i,k}[E_j|\cap_{s=1}^{j-1} E_s].
\end{displaymath}
The path of disagreement is specified by a random colouring
from $\cup_{c'\in [k]}\Omega_i(c,c')$, call this random colouring
$X$.  Let also $N_j$ be the vertices which are adjacent to the
vertex $x_j$. W.l.o.g. assume that $k>deg_i(x_j)$, for $j\leq l$.
Clearly it holds that
\begin{displaymath}
{\cal L}_{G_i,k}[E_j|\cap_{s=1}^{j-1} E_s]\leq 
\max_{\sigma \in \bigcup_{c'\in [k]}\Omega_i(c,c')}
{\cal L}_{G_i,k}[E_j|X(N_j)=\sigma_{N_j}]\leq \frac{1}{k-deg_i(x_j)}.
\end{displaymath}
Thus
\begin{displaymath}
{\cal L}_{G_i,k}[E]\leq \prod_{j=1}^{l}\frac{1}{k-deg(x_j)}
\leq {\cal P}_{G_i,k}[E].
\end{displaymath}
The lemma follows.
\end{proof}

\noindent
The following corollary is straightforward.

\begin{corollary}\label{cor:StochOrderMonotone}
Let $M=x_{1}, x_2, \ldots, x_l$ be a path in $G_i$. Let the event
$E=$''$M$ is a path of disagreement''. It holds that
\begin{displaymath}
{\cal P}_{G_i,k}[E]\leq {\cal P}_{G_{i+1},k}[E].
\end{displaymath}
\end{corollary}

\noindent
We, also, need the following lemma for the proof of Proposition
\ref{prop:DecayOfPaths}.

\begin{lemma}\label{lemma:prob-disagreementpath}
Consider the product measure ${\cal P} (G_{n,d/n},k)$, for $k\geq
(2+\epsilon)d$, for fixed $\epsilon>0$. Let $\pi$ be a permutation of $l+1$ 
vertices of $G_{n,d/n}$, for $0\leq l \leq \Theta(\log ^2 n)$. 
There exists $d_0(\epsilon)$, such that for $d>d_0(\epsilon)$ it holds that
\begin{displaymath}
{\cal P}_{G_{n,d/n},k}[\pi \textrm{ is a path of disagreement}]\leq 
\left ( \frac{d}{n}\right)^l \cdot
\left(  \left(\frac{1}{(1+\epsilon/4)d}+3n^{-0.95} \right )^l+2n^{-\log^4 n} \right).
\end{displaymath}
\end{lemma}
\begin{proof}
Call $\pi$ the path that corresponds to the permutation $\pi$, e.g.
$\pi=(x_1, \ldots x_{l+1})$. Let $\Gamma$ be an indicator variable such
that $\Gamma=1$ if $\pi$ is a path of disagreement and $\Gamma=0$,
otherwise. Let, also,  $I_{\pi}$ be the event that there exists the path 
$(x_1,\ldots,x_{l+1})$ in $G_{n,d/n}$. It holds that 
\begin{displaymath}
E_{\cal P}[\Gamma]=\left (\frac{d}{n}\right)^l 
\cdot E_{\cal P}[\Gamma| I_{\pi}].
\end{displaymath}
Let $Q_{\pi}$ denote the event that the vertices in $\pi$ have degree
less than $\log^6 n$. Using Chernoff bounds it is easy to show that
$Pr[Q_{\pi}|I_{\pi}]\geq 1-n^{-\log^4(n)}$. Also, it holds that
\begin{displaymath}
\begin{array}{lcl}
E_{\cal P}[\Gamma| I_{\pi}]
&=& E_{\cal P}[\Gamma| I_{\pi},Q_{\pi}]Pr[Q_{\pi}|I_{\pi}]+ E_{\cal P}
[\Gamma| I_{\pi},\bar{Q}_{\pi}]Pr[\bar{Q}_{\pi}|I_{\pi}]
\\ \vspace{-.3cm}\\
 &\leq& E_{\cal P}[\Gamma| I_{\pi},Q_{\pi}]+n^{-\log^4(n)}.
\end{array}
\end{displaymath}

\noindent
It suffice to show that for  $0\leq l \leq \Theta(\log^2 n) $
and sufficiently large $n$ it holds that 
\begin{equation}\label{eq:Gnpexppath}
E_{\cal P}[\Gamma | I_{\pi}, Q_{\pi}]\leq
\left( \frac{1}{(1+\epsilon/4)d}+3n^{-0.95}\right)^l.
\end{equation}
We show (\ref{eq:Gnpexppath}) by induction on $l$. Clearly for
$l=0$ the inequality in (\ref{eq:Gnpexppath}) is true. Assuming
that (\ref{eq:Gnpexppath}) holds for $l=l_0$, we will show that
it holds for $l=l_0+1$, as well.

For a vertex $w$, we let $D(w)$ denote the event that this vertex
is disagreeing. Given that all vertices in $\{x_1, \ldots,x_{l_0}\}$
are disagreeing we let $deg_{out}(x_i)$ be the number of vertices
in  $V \backslash \{x_1, \ldots, x_{l_0}\} $ that are adjacent to
$x_i$, for $1\leq i\leq l_0$. If $deg_{out}(x_i)=t$, then all the
possible subsets of $V\backslash \{x_1, \ldots, x_{l_0}\}$ with
cardinality $t$ are equiprobably adjacent to $x_i$. This implies
that 
\begin{displaymath}
Pr[x_{l_0+1} \textrm{ is adjacent to } x_i]=\frac{E[deg_{out}(x_i)]}{n-l_0} \qquad \textrm{for $0\leq i\leq l_0-1$.}
\end{displaymath}

\noindent
Let $deg_{in}(x_{l_0+1})$ be the number of neighbours of $x_{l_0+1}$ 
in $\{x_1, \ldots, x_{l_0-1}\}$.
By the linearity of expectation  we have
\begin{equation}\label{eq:expectedindegree}
E[deg_{in}(x_{l_0+1})|I_{\pi}, Q_{\pi}]\leq \frac{l_0}{n-l_0}E[deg_{out}(x_i)|I_{\pi},Q_{\pi}]\leq
n^{-0.95}.
\end{equation}
We make the simplifying assumption that if the vertex $x_{l_0+1}$
is adjacent to any vertex in $\{x_{1},\ldots, x_{l_0-1}\}$, then
it is disagreeing, regardless of the number of adjacent vertices
outside the path. By (\ref{eq:expectedindegree}) and the Markov 
inequality, we get that
\begin{displaymath}
Pr[deg_{in}(x_{l_0+1})> 0|I_{\pi}, Q_{\pi} ]\leq E[deg_{in}(x_{l_0+1})|I_{\pi}, Q_{\pi}]  \leq n^{-0.95}.
\end{displaymath}

\noindent
We denote with $E$ the event  that ``{\em $(x_1, \ldots, x_{l_0})$
is a path of disagreement, $deg_{in}(x_{l_0+1})=0$, the edge $\{x_{l_0},
x_{l_0+1}\}$ appears in $G_{n,d/n}$ and the event $Q_{\pi}$ holds}''.  
It is easy to show that $Pr[E]\geq 1-2n^{-0.95}$.
It holds that
\begin{displaymath}
\begin{array}{lcl}
\displaystyle Pr[D(x_{l_0+1})|E]&\leq&\displaystyle \sum_{j=0}^{n}Pr[D(x_{l_0+1})|E, deg_{out}(x_{l_0+1})=j]Pr[deg_{out}(x_{l_0+1})=j|E]
\\ \vspace{-.3cm}\\
&\leq &\displaystyle (1+3n^{-0.95}) \sum_{j=0}^{k-1}\frac{1}{k-j} {n \choose j}(d/n)^j(1-d/n)^{n-j}+
\\ \vspace{-.3cm}\\
&&\displaystyle +(1+3n^{-0.95})\sum_{j=k}^{n}{n \choose j}(d/n)^j(1-d/n)^{n-j}
\\ \vspace{-.3cm}\\
&\leq & q(k, d)+3n^{-0.95}
\end{array}
\end{displaymath}
where 
\begin{displaymath}
q(k,d)=\sum_{j=0}^{k-1}\frac{1}{k-j} {n \choose j}(d/n)^j(1-d/n)^{n-j}+\sum_{j=k}^{n}{n \choose j}(d/n)^j(1-d/n)^{n-j}.
\end{displaymath}
The following inequalities are straightforward.
\begin{displaymath}
\begin{array}{lcl}
q(k,d)& \leq & \displaystyle \sum_{j=0}^{k/2}\frac{1}{k-j} {n \choose j}(d/n)^j(1-d/n)^{n-j}+\sum_{j=k/2+1}^{n}{n \choose j}(d/n)^j(1-d/n)^{n-j}
\\ \vspace{-.3cm}\\
&\leq &\displaystyle \frac{2}{k}+Pr[B(n,d/n)\geq (1+\epsilon/2)d+1].
\end{array}
\end{displaymath}
Using Chernoff bounds, i.e. Corollary 2.4 from \cite{janson} we
get that 
\begin{displaymath}
Pr[B(n,d/n)\geq (1+\epsilon/2)d+1]\leq \exp(-c'd)
\end{displaymath}
where $c'=\log \epsilon-1+\frac{1}{1+\epsilon}>0$. It is clear that taking 
$k\geq (2+\epsilon)d$ for fixed $\epsilon>0$, there is  sufficiently large $d_0(\epsilon)$
such that for $d>d_0(\epsilon)$ it holds that
\begin{displaymath}
q(k,d)\leq \frac{2+2/d}{k}\leq \frac{1}{(1+\epsilon/4)d}.
\end{displaymath}
The lemma follows.
\end{proof}

\begin{propositionproof}{\ref{prop:DecayOfPaths}}
Let the event $B=$``$v_i$ and $u_i$ are connected through a path
of disagreement of length at most $\log^2 n$''. Also, let the event
$C=$``$v_i$ and $u_i$ are connected through a path of length greater
than $\log^2 n$''. Clearly it holds that
\begin{displaymath}
{\cal L}_{G_i,k}[A]\leq {\cal L}_{G_i,k}[B]+{\cal L}_{G_i,k}[C],
\end{displaymath}
where $L_{G_i,k}$ is the probability distribution we defined at
the begining of this section. When there is no danger of confusion
we drop the subscript $G_i, k$
The proposition will follow by calculating the probabilities
${\cal L}[B]$ and ${\cal L}[C]$.

Consider an enumeration of all the permutations of $l$ vertices
in $G_{i}$ with first the vertex $v_i$ and last the vertex $u_i$.
Let $\pi_0(l),\pi_1(l),\ldots$ be the permutations in the order
they appear in the enumeration. Let $\Gamma_j(l)$ be the random
variable such that
\begin{displaymath}
\Gamma_j(l)=\left\{
\begin{array}{lcl}
1 &\qquad & \textrm{the path that corresponds to $\pi_i(l)$ is a path of disagreement}\\
0 && \textrm{otherwise}.
\end{array}
\right.
\end{displaymath}
Let, also, $\Gamma(l)=\sum_j \Gamma_j(l)$. It is easy to see that
the number of sumads in the previous sum are at most $n^{l-1}$. 
Towards computing ${\cal L}(C)$, we need to calculate the following
expectation
$$E_{\cal L}\left [\sum_{l= l_0}^{\log^2n}\Gamma(l)\right ],$$
where $l_0=\frac{\log n}{9\log d}$. 
However, we have to take into consideration that we have conditioned
that $v_i$ and $u_i$ are at distance at least $\frac{\log n}{9\log d}$.
To this end, it is direct to show that if $Z$ the number of paths
of length at most $\frac{\log n}{9\log d}-1$ between two vertices
of $G_{i}$, then 
\begin{displaymath}
E[Z]\leq \sum_{l\leq \frac{\log n}{9\log d}-1}n^{l-1}\left(\frac{d}{n} \right)^l\leq n^{-9/10}.
\end{displaymath}
Thus, letting $\hat{p}$ be the probability of the event that two vertices
are at distance is at least $\frac{\log n}{9\log d}$ the Markov inequality
suggests that $\hat{p} \geq 1-n^{-9/10}$.
Using Lemma \ref{lemma:StochOrderSimple}, Lemma \ref{lemma:prob-disagreementpath}
and Corollary \ref{cor:StochOrderMonotone} we get that
\begin{displaymath}
\begin{array}{lcl}
\displaystyle E_{\cal L}\left[\sum_{l=l_0}^{\log^2n}\Gamma(l) \right]
&\leq& \displaystyle  \hat{p}^{-1}\sum_{l= l_0}^{\log^2n} n^{l-1}
\left(\frac{d}{n} \right)^{l}\left(  \left( \frac{1}{(1+\epsilon/4)d}+3n^{-0.95} \right )^l+2n^{-\log^4 n} \right)
\\ \vspace{-.3cm}\\
&\leq& \displaystyle  \frac{1}{n\hat{p}} 
\sum_{l=l_0}^{\log^2n}
\left( \left( (1+\epsilon/4)^{-1}+3dn^{-0.95} \right )^l+2d^{l}n^{-\log^4 n} \right).
\end{array}
\end{displaymath}
Note that $d^ln^{-\log^4 n}=O(n^{-\log^4n})$, for $l=O(\log^2 n)$.
Thus, for sufficiently large $n$ and $d$ we get that
\begin{displaymath}
\begin{array}{lcl}
\displaystyle E_{\cal L}\left[\sum_{l=l_0}^{\log^2n} \Gamma(l) \right]
&\leq& \displaystyle \sum_{l= l_0}^{\log^2n}\displaystyle \frac{3}{2n}\left(1+\frac{\epsilon}{5}\right)^{-l}
\\ \vspace{-.3cm}\\
&\leq & \displaystyle \frac{3}{2n}\left(1+\frac{\epsilon}{5}\right)^{-l_0}\frac{1}{1-(1+\epsilon/5)^{-1}}\leq
\frac{15+3\epsilon}{2\epsilon}n^{-\left(1+\frac{\epsilon}{45\log d} \right)}.
\end{array}
\end{displaymath}
Using the Markov inequality we get that
\begin{equation}\label{eq:ProbB}
{\cal L}[B]\leq E_{\cal L}\left[\sum_{l\geq l_0}^{\log^2n} \Gamma(l) \right]
\leq \frac{15+3\epsilon}{2\epsilon}n^{-\left(1+\frac{t}{45\log d} \right)}.
\end{equation}

Let $P(l)$ be the number of paths of disagreement between $v_i$ and
{\em any} vertex of $G_i$, that have length $l$. It is direct that
\begin{displaymath}
{\cal L}[C]\leq Pr\left[P(\log^2n)>0\right].
\end{displaymath}
The above inequality follows by noting that so as to have a path of
disagreement connecting $v_i$ and $u_i$ which has length at least $l$, 
we should have some path of disagreement of length $l$ leaving $v_i$.
Using Markov's inequality we get that 
\begin{displaymath}
\begin{array}{lcl}
Pr\left[P(\log^2n)>0\right]&\leq& \displaystyle  E_{\cal L}\left[P(\log^2n) \right]
\\ \vspace{-.3cm} \\
&\leq & \displaystyle 
\hat{p}^{-1}n^{\log^2n-1} \left(\frac{d}{n} \right)^{\log^2n}
\left(\left( \frac{1}{(1+\epsilon/4)d}+3n^{-0.95} \right )^{\log^2n}+2n^{-\log^4 n} \right)
\\ \vspace{-.3cm} \\
&\leq & \displaystyle \frac{1}{\hat{p}n} \left( 1+\epsilon/5 \right )^{-\log^2n}=\Theta\left(n^{-\frac{\epsilon}{10}\log n}\right).
\end{array}
\end{displaymath}
The proposition follows.
\end{propositionproof}

\section{Proofs}
\inappendix{
\subsection{Lemma \ref{lemma:BijectionTVD}}\label{sec:lemma:BijectionTVD}

\begin{proof}
Let $x$ be a r.v. distributed as in $\nu$. The proof of this lemma
is going to be made by coupling $x$ and $z'$. In particular, we
show that there is a coupling of $x$ and $z'$ such that 
\begin{displaymath}
Pr[x\neq z']\leq \alpha.
\end{displaymath}
Then the lemma will follow by using Coupling Lemma \cite{coupling-lemma}.
Let $(\Omega'_1, \Omega'_2)$ be the isomorphic pair of the 
$\alpha$-isomorphism between $\Omega_1$ and $\Omega_2$.
Observe that $|\Omega'_i|\geq (1-\alpha)|\Omega_i|$, for $i=1,2$.
Also, it holds that 
\begin{displaymath}
Pr[z'=\sigma|z\in \Omega'_1]=\frac{1}{|\Omega'_2|}  \qquad \forall \sigma \in \Omega'_2.
\end{displaymath}
Note that when we restrict the input of the $\alpha$-function $H$
only to members of $\Omega'_1$, then $H$ is by definition a 
bijection  between the sets $\Omega'_1$ and $\Omega'_2$. The above 
equality then follows by using Corollary \ref{corollary:IsoSmpl} and 
noting that conditional on the fact that $z\in \Omega'_1$, $z$ is 
distributed uniformly over $\Omega'_1$. 
Also it is easy to get that
\begin{displaymath}
Pr[x=\sigma |x\in \Omega'_2]=\frac{1}{|\Omega'_2|} \qquad \forall \sigma \in \Omega'_2.
\end{displaymath}
Let
\begin{equation}\label{eq:ProbOfE}
p=\min\{Pr[x\in \Omega'_2], Pr[z\in \Omega'_1]\}\geq 1-\alpha.
\end{equation}
The above inequality follows from the assumption that $\Omega_1$ and $\Omega_2$
are $\alpha$-isomorphic.

It is clear that we can have a coupling between $x$, $z$ and $z'$
such that the event $E=$``$z\in \Omega'_1$ and $x\in \Omega'_2$''
holds with probability $p$ (see (\ref{eq:ProbOfE})). In this
coupling, if the event $E$ holds, then $x$ and $y$ are distributed
uniformly over $\Omega'_2$ and $\Omega'_1$, respectively. 
This means that we can make an extra arrangement such that when
$E$ holds to have $x=H(z)$, as well. Since $z'=H(z)$, it is direct
that when the event $E$ holds we, also, have that $x=z'$. We conclude
that it the above coupling it holds that $Pr[x=z']\geq Pr[E]$. Thus,
$$Pr[x\neq z']\leq 1- Pr[E]=1-p=\alpha.$$
The lemma follows.
\end{proof}

\subsection{Lemma \ref{lemma:isomorphism}}\label{sec:lemma:isomorphism}
\begin{proof}
First we are going to show that for any $\sigma \in S(c,c)$, it
holds that $H(\sigma,q)$ is a proper colouring of $G$. Assume the
contrary, i.e. that there is $\sigma \in S(c,c)$ such that
$H(\sigma, q)$ is a non-proper colouring, i.e. there is a monochromatic
edge $e$.
Let $Q_{\sigma_v,q}$ be the disagreement graph specified by 
$\sigma$ and $q$. It is direct that the monochromatic edge is
either incident to two vertices in $Q_{\sigma_v, q}$ or to
some vertex in $Q_{\sigma_v, q}$ and some vertex outside the
disagreement graph.

It is direct that $H(\sigma,q)$ does not cause any monochromatic
edge between two vertices in $Q_{\sigma_v,q}$. To see this, note
that the disagreement graph is bipartite and $\sigma$ specifies
exactly one colour for each part of the graph, while $H(\sigma,q)$
switches the colours of the two parts.
On the other hand, $H(\sigma,q)$ cannot cause any monochromatic
edge between a vertex in $Q_{\sigma_v,q}$ and some vertex outside
the disagreement graph.
This follows by the fact that the disagreement graph is maximal.
Thus, there is no edge $w$ outside $Q_{\sigma_v,q}$ such that 
$\sigma_w\in \{q,\sigma_v\}$ while at the same time $w$ is adjacent
to some vertex in $Q_{\sigma_v,q}$.

Also, it is direct to show that for any $\sigma\in S(c,c)$,
it holds that $H(\sigma,q)\in S(q,c)$.  This follows by the
definition of the sets $S(c,c)$ and $S(c,q)$. It remains to 
show that $H(\cdot,q):S(c,c)\to S(q,c)$ is a bijection.

We show that $H(\cdot, q)$ has range the set $S(q,c)$, i.e. it is
{\em surjective} map,
ie. for any $\sigma \in S(q,c)$ there is $\sigma'\in S(c,c)$ 
such that $\sigma=H(\sigma',q)$. It is direct to see that such 
$\sigma'$ exists, moreover, it holds $\sigma'=H(\sigma,c)$.

Finally, we need to show that $H(\cdot, q)$ is {\em one-to-one}, 
i.e. there are no two $\sigma_1,\sigma_2\in S(c,c)$ such that
$H(\sigma_1,q)=H(\sigma_2,q)$. Using arguments similar to those
in the previous paragraph it is direct to see that there cannot
be such a pair of colourings. 

Thus, since $H(\cdot,q):S(c,c)\to S(q,c)$ is surjective and
one-to-one it is a bijection. The lemma follows.
\end{proof}

\subsection{Theorem \ref{theorem:STEPAccuracy}}\label{sec:theorem:STEPAccuracy}

\begin{proof}
Let $X$ be the input of {\tt STEP}, i.e. a random $k$-colouring
of $G$. Let $Y$ be equal to the colouring that is returned by the
algorithm. Also, let $Z$ be a random variable distributed as in
$\nu$. The proof of the theorem is going to be made by coupling
$Z$ and $Y$ and by showing that in this coupling it holds that
\begin{displaymath}
Pr[Z\neq Y]\leq \alpha.
\end{displaymath}

\noindent
The reader should observe that for any $q,c\in [k]$ such that 
$c\neq q$, it holds that
\begin{equation}\label{eq:equationA}
Pr[Z(v)=q|Z(u)=c]=Pr[X(v)=q|X(u)=c,X(v)\neq c]=\frac{1}{k-1}
\end{equation}
and 
\begin{equation}\label{eq:equationB}
Pr[X(v)=X(u)=c|\textrm{$X$ is bad }]=\frac{1}{k},
\end{equation}
due to symmetry. Also, it is direct to show that
\begin{equation}\label{eq:equationC}
Pr[Y(v)=q|X(u)=c]=\frac{1}{k-1}
\end{equation}
for every $q\in [k]\backslash\{c\}$.
Now we are going to construct the coupling. We need to involve the
variable $X$, the input of {\tt STEP}, in this coupling. First, set
$Z(u)=X(u)$ and then set $Z(v)=Y(v)$. Using the above observations
it is straightforward to show that $Z(u)$ and $Z(v)$ are set,
respectively, according to the appropriate distribution (due to
(\ref{eq:equationA}) and (\ref{eq:equationC})).

We reveal the values of $X(v)$, $X(u)$ and $Y(v)$. By the above
coupling we also have the values of $Z(v)$ and $Z(u)$. We consider
two cases, depending on whether $X$ is a good or a bad colouring.

If $X(v)\neq X(u)$, i.e. $X$ is good, then we have $X=Y$ and we
can set directly $X=Z$. Thus, for the coupling it holds
\begin{displaymath}
Pr[Y\neq Z|\textrm{$X$ is good}]=0.
\end{displaymath}

\noindent
If $X(u)=X(v)$, then w.l.o.g. we can
assume $X(u)=X(v)=c$, for some $c\in [k]$. In this case, we
choose whether $X\in S(c,c)$ or not. For this choice, the 
Assumption \ref{assm:isomorphism} suggests that
\begin{displaymath}
Pr[X\in S(c,c)|X(u)=X(v)=c]\geq 1-\alpha. 
\end{displaymath}
Similarly for $Z$, assume that $Z(u)=c$ and $Z(v)=q$, with $c\neq
q$. Again Assumption \ref{assm:isomorphism}  suggests that 
\begin{displaymath}
Pr[Z\in S(q,c)|X(u)=X(v)=c]\geq 1-\alpha. 
\end{displaymath}
Let the event $E=$``$X\in S(c,c)$ and $Z\in S(q,c)$''. Having set
$X(v),X(u)$, $Z(v)$, $Z(u)$, $Y(v)$, the two previous inequalities
suggest that we can couple  $X$ and $Z$ such that the probability
of the event $E$ to occur is at least $1-\alpha$.

\begin{claim}\label{claim:YUniformE}
Conditional on the event $E$, $Y$ is distributed uniformly over
$S(q,c)$.
\end{claim}
Conditional on the event $E$, it is easy to observe that $Z$ is,
also, distributed uniformly over $S(q,c)$. This observation and
Claim \ref{claim:YUniformE} suggest that 
\begin{displaymath}
Pr[Z\neq Y|E]=0.
\end{displaymath}
Gathering all the above together and applying the law of total
probability we get the following for the coupling:
\begin{displaymath}
Pr[Z\neq Y]\leq Pr[Z\neq Y|\textrm{$X$ is good}]+Pr[Z\neq Y|E]
+Pr[\bar{E}]\leq \alpha.
\end{displaymath}
The theorem follows.
\end{proof}

\begin{claimproof}{\ref{claim:YUniformE}}
Conditional on the event $E$, the random variable $X$ is distributed
uniformly over $S(c,c)$.  Note that $S(c,c)$ and $S(q,c)$ are isomorphic,
due to Lemma \ref{lemma:isomorphism}. The same lemma suggests that 
we can have a bijection between the two isomorphic sets by 
taking $H(\cdot ,q)$ and restricting its input only to colourings
in $S(c,c)$. Thus, since $X$ is distributed uniformly over $S(c,c)$,
$H(X,q)=Y$ is distributed uniformly over $S(q,c)$, by Corollary 
\ref{corollary:IsoSmpl}. he claim follows.
\end{claimproof}

\subsection{Lemma \ref{lemma:StepAc}}\label{sec:lemma:StepAc}
\begin{proof}
Note that the time complexity of computing the value of $H(\sigma, q)$
is dominated by the time we need to reveal the disagreement graph 
$Q_{\sigma_v,q}$, for some $q\in [k]$. We show that we need $O(|E|)$ 
steps to reveal the disagreement graph $Q_{\sigma_v,q}$.

We can reveal the graph $Q_{\sigma_v,q}$ in steps $j=0,\ldots, |E|$, 
where $E$ is the set of edges of $G$.
At step $0$ the disagreement graph $Q_{\sigma_v,q}(0)$ contains only
the vertex $v$. Given the graph $Q_{\sigma_v,q}(j)$ we construct
$Q_{\sigma_v,q}(j+1)$ as follows: Pick some edge which is incident
to a vertex in $Q_{\sigma_v,q}(j)$. If the other end of this edge 
is incident to a vertex outside $Q_{\sigma_v,q}(j)$ that is coloured
either $\sigma_v$ or $q$ then we get $Q_{\sigma_v,q}(j+1)$ by inserting
this edge and the vertex into $Q_{\sigma_v,q}(j)$. Otherwise 
$Q_{\sigma_v,q}(j+1)$ is the same as $Q_{\sigma_v,q}(j)$. We never
pick the same edge twice. 

It is direct to show that in the above procedure it holds that
$Q_{\sigma_v,q}=Q_{\sigma_v,q}(|E|)$. Thus. the time complexity
of a $q$-switching of a given colouring of $G$ is $O(|E|)$.
\end{proof}

\subsection{Theorem \ref{thrm:Accuracy}}\label{sec:thrm:Accuracy}

\begin{proof}
Let $X_i$ be a random variable which is distributed uniformly over
$\Omega_i$, $i=0,\ldots, r$. It suffices to provide a coupling of 
$X_r$ and $Y_r$, such that $$Pr[X_r\neq Y_r]\leq r\cdot \alpha.$$

\noindent
%
%
Working as in the proof of Theorem \ref{thrm:Accuracy} we get
the following: There is a coupling of $X_{i}, X_{i+1}$ such that
for the event 
$E_{i}=``$ $X_{i}$ is good or there are $c,q\in [k]$ such that
$X_{i}\in S_i(c,c)$ and $X_{i+1}\in S(q,c)$ it holds that 
\begin{displaymath}
Pr[E_i]\geq 1-\alpha.
\end{displaymath}

\noindent
Now consider the random variables $Z=(X_{0}, X_{1}, \ldots, X_{r-1})$
and $Z'=(X_{1}, X_{2}, \ldots X_{r})$ and $W=(Y_1,\dots,Y_r)$.
Consider, also, the event $E=\cap_{i=0}^{r-1}E_i$, where $E_{i}$
is the event defined above.  All the above discussion suggests 
two facts: 
First, there is a coupling between $Z$, $Z'$ and $W$ such that 
\begin{displaymath}
Pr[E]\geq 1-r\alpha.
\end{displaymath}
Second, if in this coupling the event $E$ occurs we can have $Z'=W$, 
i.e. $X_{i}=Y_{i}$, for $i=1,\ldots,r$. 
To see thus, consider the following:
If the event $E_i$ occurs we can have either $X_i=X_{i+1}$ or
$X_{i+1}=H(X_i,q)$ for appropriate $q$. When $E$ occurs, we 
have this property for all $i=0,\ldots,r-1$. A direct inductive
argument implies $Z'=W$.
The theorem follows by noting that  $$Pr[Y_r\neq X_r]\leq1-Pr[E].$$
\end{proof}

\subsection{Lemma \ref{lemma:seqsubprop}}\label{sec:lemma:seqsubprop}
\begin{proof}
For (1) it suffice to show that with probability at least $1-n^{-2/3}$
all the cycles of length less than $\frac{\log n}{9\log d}$ in $G_{n,d/n}$
do not share edges with each other. Let $\gamma=(9\log(d))^{-1}$. Assume
the opposite, there are at least two cycles, each of length at least
$\gamma \log n$ that intersect with each other. Then, there  must
exist a subgraph of $G_{n,d/n}$ that contains at most $2\gamma
\log n$ vertices  while the number of edges exceeds  by 1, or more,
the number of vertices.

Let $D$ be the event that in $G_{n,d/n}$ there exists a set of $r$
vertices which have $r+1$ edges between them. For $r \leq 2\gamma
\log n$ we have the following:
\begin{displaymath}
\begin{array}{lcl}
Pr[D] & \leq & \displaystyle \sum_{r=1}^{\gamma\log n} {n \choose r} { {r \choose 2} \choose r+1} (d/n)^{r+1} (1-d/n)^{{r \choose 2}-(r+1)}
\\ \vspace{-.3cm} \\
& \leq & \displaystyle \sum_{r=1}^{\gamma \log n} \left( \frac{n e}{r} \right )^r \left ( \frac{r^2 e}{2(r+1)} \right )^{r+1} (d/n)^{r+1} 
\leq 
\frac{e \cdot d}{2n}\sum_{r=1}^{\gamma \log n} \left ( \frac{e^2 d}{2} \right )^{r}
\\ \vspace{-.3cm}\\
& \leq & 
\displaystyle \frac{C}{n} \left (\frac{e^2d}{2} \right)^{2\gamma  \log n}.
\end{array}
\end{displaymath}
Having $2\gamma \cdot \log(e^2d/2)< 1 $, the quantity in the
r.h.s. of the last inequality is $o(1)$, in particular it is of
order $\Theta(n^{\gamma \log (e^2d/2)-1})$. Thus, for 
$\gamma =(9\log d)^{-1}$  there is no connected component that 
contains two cycles with probability at least $1-2n^{-2/3}$.

If we include in $G_0$ all the edges that belong to small cycles,
i.e. of length less than $\frac{\log n}{9\log d}$ then it is straightforward
that (2) holds.

For (3), we let $E(G_{n,d/n})$ be the number of edges in $G_{n,d/n}$.
Using standard probabilistic tools, i.e. Chernoff bounds, it is direct
to get that
\begin{displaymath}
Pr\left [E(G_{n,d/n})\geq (1+n^{-1/3})\frac{dn}{2} \right]\leq \exp\left(-n^{1/4} \right).
\end{displaymath}
It is direct that $r$, the number of terms in the sequence of
subgraphs of $G_{n,d/n}$, is upper bounded by $E(G_{n,d/n})$.
Thus, the above inequality implies that
$$
Pr\left[r\geq (1+n^{-1/3})\frac{dn}{2}\right]\leq \exp\left(-n^{1/4} \right).
$$
The lemma follows.
\end{proof}

\subsection{Proof of Corollary \ref{corollary:IsoSmpl}}\label{sec:corollary:IsoSmpl}
The existence of the bijection $T$ implies that $|\Omega_1|=|\Omega_2|$. Thus $\forall 
\xi \in \Omega_1$ it holds that
\begin{displaymath}
Pr[X=\xi]=Pr[T(X)=T(\xi)]=\frac{1}{|\Omega_1|}.
\end{displaymath}
Since, for every $\sigma\in \Omega_2$ there is a unique $\sigma'\in \Omega_1$
such that $T(\sigma')=\sigma$ we get that
\begin{displaymath}
Pr[T(X)=\sigma]=\frac{1}{|\Omega_1|}=\frac{1}{|\Omega_2|}.
\end{displaymath}
The corollary follows.


}

\end{document}